\documentclass[lettersize,journal]{IEEEtran}
\IEEEoverridecommandlockouts
\usepackage{cite}
\usepackage{amsmath,amssymb,amsfonts}
\usepackage[amsmath,thmmarks]{ntheorem}
\usepackage{algorithmic}
\usepackage{graphicx} 
\usepackage{epsfig} 
\usepackage{epstopdf} 
\usepackage{textcomp}
\usepackage{xcolor}

\newenvironment{proof}{{\indent \indent \it Proof:}}{\hfill $\blacksquare$\par}

{\theoremheaderfont{\itshape}
\theorembodyfont{\upshape}
\theoremseparator{:}
\newtheorem{proposition}{\hspace{1em}Proposition}}

{\theoremheaderfont{\itshape}
\theorembodyfont{\upshape}
\theoremseparator{:}
\newtheorem{corollary}{\hspace{1em}Corollary}}

{\theoremheaderfont{\itshape}
\theorembodyfont{\upshape}
\theoremseparator{:}
\newtheorem{theorem}{\hspace{1em}Theorem}}

{\theoremheaderfont{\itshape}
\theorembodyfont{\upshape}
\theoremseparator{:}
\newtheorem{lemma}{\hspace{1em}Lemma}}

{\theoremheaderfont{\itshape}
\theorembodyfont{\upshape}
\theoremseparator{:}
\newtheorem{remark}{\hspace{1em}Remark}}
\usepackage{indentfirst}
\usepackage{array}
\usepackage{bm}
\usepackage[caption=false,font=normalsize,labelfont=sf,textfont=sf]{subfig}
\usepackage{textcomp}
\usepackage{stfloats}
\usepackage{url}
\usepackage{verbatim}
\usepackage{graphicx}
\usepackage{caption}
\usepackage{amssymb}
\usepackage{float}
\usepackage{balance}
\usepackage{diagbox}
\usepackage{booktabs}

\allowdisplaybreaks[4]

\makeatletter

\newcommand{\Rmnum}[1]{\expandafter\@slowromancap\romannumeral #1@}
\makeatother

\def\BibTeX{{\rm B\kern-.05em{\sc i\kern-.025em b}\kern-.08em
    T\kern-.1667em\lower.7ex\hbox{E}\kern-.125emX}}

\usepackage {setspace}
\makeatletter
\newcommand\notsotiny{\@setfontsize\notsotiny\@vipt\@viipt}
\makeatother

\usepackage[explicit]{titlesec}
 
\titlespacing*{\section}{0pt}{1.2ex plus .0ex minus .0ex}{.3ex plus .0ex}
\titlespacing*{\subsection}{0pt}{1.2ex plus .0ex minus .0ex}{.3ex plus .0ex}

\begin{document}
\title{Secure Wireless-Powered zeRIS Communications
}

\author{Jingyu Chen, Kunrui Cao, Panagiotis D. Diamantoulakis,~\IEEEmembership{Senior Member,~IEEE}, \\Lu Lv, Liang Yang, Haolian Chi, and Haiyang Ding
\thanks{
    Jingyu Chen, Kunrui Cao, Haolian Chi, and Haiyang Ding are with the School of Information and Communications, National University of Defense Technology, Wuhan 430035, China (e-mail: chenjingyu@nudt.edu.cn; krcao@nudt.edu.cn; chihaolian20@nudt.edu.cn; dinghy2003@nudt.edu.cn).
    
    Panagiotis D. Diamantoulakis is with the Department of Electrical and Computer Engineering, Aristotle University of Thessaloniki, Thessaloniki 54124, Greece (padiaman@auth.gr).

    Lu Lv is with the State Key Laboratory of Integrated Services Networks, Xidian University, Xi'an 710071, China (e-mail: lulv@xidian.edu.cn).

    Liang Yang is with the College of Computer Science and Electronic Engineering, Hunan University, Changsha 410082, China (e-mail: liangy@hnu.edu.cn).
    }
    \vspace{-0.8cm}
}
\maketitle

\begin{abstract}
    This paper introduces the concept of wireless-powered zero-energy reconfigurable intelligent surface (zeRIS), and investigates a wireless-powered zeRIS aided communication system in terms of security, reliability and energy efficiency.
    In particular, we propose three new wireless-powered zeRIS modes:
    1) in mode-\Rmnum{1}, $N$ reconfigurable reflecting elements are adjusted to the optimal phase shift design of information user to maximize the reliability of the system;
    2) in mode-\Rmnum{2}, $N$ reconfigurable reflecting elements are adjusted to the optimal phase shift design of cooperative jamming user to maximize the security of the system;
    3) in mode-\Rmnum{3}, $N_1$ and $N_2$ $(N_1+N_2=N)$ reconfigurable reflecting elements are respectively adjusted to the optimal phase shift designs of information user and cooperative jamming user to balance the reliability and security of the system.
    Then, we propose three new metrics, i.e., joint outage probability (JOP), joint intercept probability (JIP), and secrecy energy efficiency (SEE), and analyze their closed-form expressions in three modes, respectively.
    The results show that under high transmission power, all the diversity gains of three modes are 1, and the JOPs of mode-\Rmnum{1}, mode-\Rmnum{2} and mode-\Rmnum{3} are improved by increasing the number of zeRIS elements, which are related to $N^2$, $N$, and $N_1^2$, respectively.
    In addition, mode-\Rmnum{1} achieves the best JOP, while mode-\Rmnum{2} achieves the best JIP among three modes.
    We exploit two security-reliability trade-off (SRT) metrics, i.e., JOP versus JIP, and normalized joint intercept and outage probability (JIOP), to reveal the SRT performance of the proposed three modes.
    It is obtained that mode-\Rmnum{2} outperforms the other two modes in the JOP versus JIP, while mode-\Rmnum{3} and mode-\Rmnum{2} achieve the best performance of normalized JIOP at low and high transmission power, respectively.
\end{abstract}

\begin{IEEEkeywords}
    Wireless-powered communication, zero-energy reconfigurable intelligent surface, security-reliability trade-off, physical layer security.
\end{IEEEkeywords}

\section{Introduction}
The use of reconfigurable intelligent surfaces (RISs), as a groundbreaking paradigm, is expected to realize reliable and secure communications with massive connectivity, higher network coverage, greater network capacity and higher energy efficiency in the era of 6G, thanks to a large number of low-power reflecting elements \cite{9424177}.
Moreover, RIS can enhance the desired signal or suppress detrimental interference by manipulating the amplitude and phase of signal to superimpose or cancel signal, thus effectively controlling the wireless propagation environment \cite{9326394}.

Due to the broadcast nature of wireless channels, wireless communications are extremely vulnerable to malicious eavesdropping attacks \cite{10605790}.
Physical layer security (PLS) is an appealing approach to strengthen the information security, which exploits the intrinsic stochastic nature of wireless channels to guarantee that an eavesdropper gleans no useful information at the physical layer about the legitimate information \cite{10044975}.
It is noteworthy that RIS can effectively improve the PLS, since it controls the wireless propagation environment to enhance the transmissions of legitimate information, and attenuate the transmission of eavesdropping signals.
Among the various techniques employed in the context of PLS, RIS aided cooperative jamming stands out as a prominent method \cite{10188924}. 
However, additional artificial noise leads to an increase of energy cost at cooperative nodes.
This is mitigated by energy harvesting (EH) technology.

Wireless-powered communication (WPC), combining the technologies of microwave wireless power transfer and radio frequency EH, exploits a harvest-then-transmit protocol to divide up the downlink energy transfer (ET) and uplink information transmission (IT) \cite{10414117}.

As the number of RIS elements increases, the energy consumption of RIS-aided WPC systems will be comparable to that of the radio frequency power, which needs to be considered into the overall communication design. 
Although the RIS can be powered by battery or grid, there are some drawbacks: 1) Powering the RIS through conventional power line networks not only increases the implementation cost, but also decreases the RIS deployment flexibility; 
2) Batteries typically have limited energy storage, thus significantly limiting the service life of the RIS;
3) Due to the complexity of the environment, manually replacing RIS batteries is costly, or even impossible.
To this end, some recent works have delved into the innovative concept of transforming RIS into zero-energy devices (ZEDs) which have emerged as a prominent design for 6G green communication architecture \cite{9679387,10098556}.
This transformation is seen as a pivotal step towards achieving enhanced performance with a focus on energy efficiency \cite{9895266}. 
Subsequently, the concept of zero-energy RIS (zeRIS) is introduced by equipping RIS elements with EH modules, which is completely self-sufficient in terms of energy by relying on ambient electromagnetic waves for its energy \cite{10348506}. 
In contrast to traditional RIS, zeRIS with EH adheres to the fundamental principle of sustainability by eliminating the reliance on dedicated power sources \cite{10521637}.

Driven by the recent advances in zeRIS, we in this paper introduce the concept of wireless-powered zeRIS which achieves a self-sustainable zeRIS and its information transmission via WPC. 
Recalling PLS of information transmission, as for secure wireless-powered zeRIS communications, two fundamental questions arise:
\textit{1) During ET stage, can the harvested energy at zeRIS support its stable operation?
2) If so, during IT stage, is the use of wireless-powered zeRIS helpful to improve the reliability, security and energy efficiency of communications?}
To our best knowledge, these two questions have not been clearly articulated in the literature.
Although initial works \cite{9679387,9703335,10098556,9895266,10348506,10521637,10083178,10113228,10370741} focused on the integration of WPC and RIS, their results are not applicable to secure wireless-powered zeRIS.
This is because the energy outage of zeRIS in the ET stage is not considered in the previous literature, and if the energy transfer is insufficient, there will occur an energy outage event that zeRIS can not operate stably in the IT stage.
Moreover, the previous literature did not consider the trade-off between security, reliability and energy efficiency, and there is no standardized metric to measure them.
In addition, the study on enhancing the transmissions of both legitimate information and artificial noise by segmenting zeRIS into different regions has never been presented.
This paper is motivated by the aforementioned considerations, and aims to bridge this gap.
The contributions of this paper are summarized as follows.

\begin{itemize}
    \item We introduce the concept of wireless-powered zeRIS, and investigate a wireless-powered zeRIS communication system in terms of security, reliability and energy efficiency.
    In addition to information user, an idle user in the system acts as a friendly jamming user to emit artificial noise with the energy harvested from power station.
    By configuring different numbers of reflecting elements at zeRIS, we propose three new wireless-powered zeRIS modes:
    1) in mode-\Rmnum{1}, all reconfigurable reflecting elements are adjusted to the optimal phase shift design of information user to maximize the reliability of the system;
    2) in mode-\Rmnum{2}, all reconfigurable reflecting elements are adjusted to the optimal phase shift design of friendly jamming user to maximize the security of the system;
    3) in mode-\Rmnum{3}, reconfigurable reflecting elements are divided into two parts that are adjusted to the optimal phase shift designs of information user and friendly jamming user, respectively, to balance the reliability and security of the system.
        
    \item To comprehensively evaluate the performance of proposed three new wireless-powered zeRIS modes in terms of reliability, security, and energy consumption, we propose three new metrics, i.e., joint outage probability (JOP), joint intercept probability (JIP), and secrecy energy efficiency (SEE).
    Moreover, we derive their closed-form expressions, which match well with the simulation results, verifying the correctness of our theoretical analysis.
    To analyze the security-reliability trade-off (SRT) performance of the proposed three new wireless-powered zeRIS modes, we exploit two SRT metrics, i.e., JOP versus JIP, which refers to treating JOP as a function of JIP, and normalized joint intercept and outage probability (JIOP), which is the mean of the sum of JOP and JIP.
  
    \item Theoretical and numerical results reveal the following conclusions:
    1) Under high transmission power, all the diversity gains of three modes are 1, and the JOPs of mode-\Rmnum{1}, mode-\Rmnum{2} and mode-\Rmnum{3} are improved by increasing the number of zeRIS elements, which are related to $N^2$, $N$, and $N_1^2$ ($N_1$ is the number of optimal phase shift elements of information user), respectively;
    2) From the perspective of reliability, mode-\Rmnum{1} achieves the best JOP, while mode-\Rmnum{2} achieves the best JIP from the perspective of security;
    3) From the holistic perspective of reliability and security, mode-\Rmnum{2} outperforms the other two modes in the analysis of JOP versus JIP, while mode-\Rmnum{3} and mode-\Rmnum{2} achieve the best performance of normalized JIOP at low and high transmission power, respectively;
    4) Interestingly, mode-\Rmnum{3} achieves the lowest normalized JIOP with increasing time allocation factor, and the highest SEE with increasing transmission power among three modes.
\end{itemize}

\section{System Model}

\begin{figure}[t]
    \centering
    \includegraphics[width=0.8\linewidth]{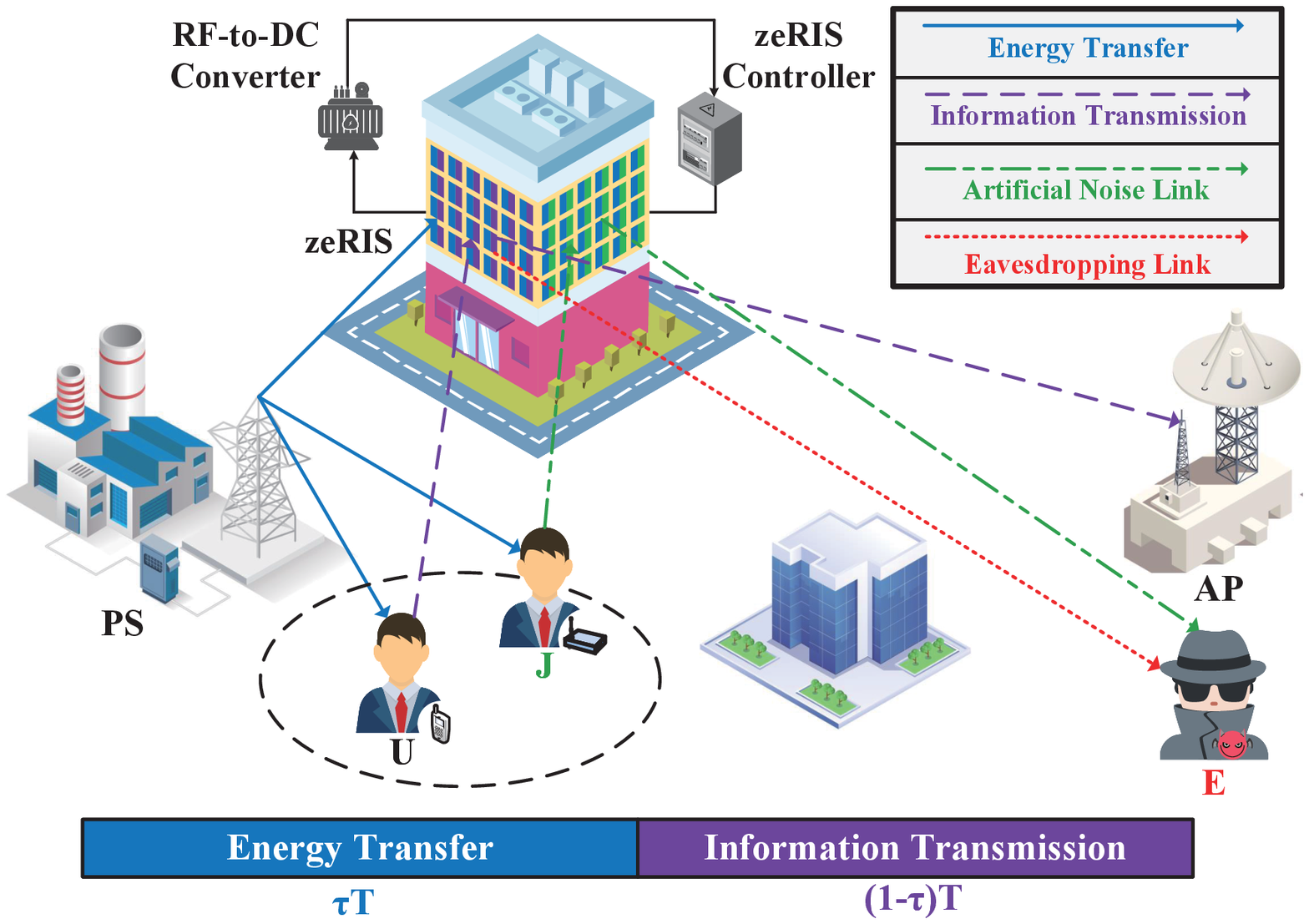}
    \captionsetup{font={footnotesize}} 
    \caption{System model}
    \label{fig1}
\end{figure}

As depicted in Fig. \ref{fig1}, we consider a wireless-powered zeRIS communication system consisting of a power station (PS), an access point (AP), an energy-constraint user (U), a malicious eavesdropper (E), and a friendly jamming user (J).
Each node in the system is equipped with a single antenna, and operates in a half-duplex mode.
The zeRIS is configured by $N$ reflecting elements with tunable impedance and a controller that takes charge of the configuration of the reflecting elements.
Specifically, each zeRIS element includes an RF-to-DC converter, which can convert a predefined fraction of the incident wave power into a DC voltage.
This voltage is stored in a capacitor as short-term energy that meets the initial energy requirement for zeRIS operation \cite{10348506}. 
Within a time frame $T$, the communication process is divided into two stages, i.e., ET stage and IT stage.
In the ET stage, PS transfers radio frequency (RF) energy signals to U, J and zeRIS.
In the IT stage, U uses the harvested RF energy to send information signal to AP, while J simultaneously emits the artificial noise to mask the transmission of U. 
The zeRIS reflects the incident signals of U and J to enhance the received power at AP and deteriorate the reception of E, respectively.
In the process, the zeRIS achieves self-sufficient modulation of controllable reflecting elements by using harvested energy.
It is assumed that the direct communication links between U/J and AP/E are hindered by physical obstructions.

The channel coefficients between PS and U/J are denoted by $h_{pu}$ and $h_{pj}$, respectively.
The channel vectors between PS/U/J and zeRIS are denoted by ${\bf{h}}_{pr} \in \mathbb{C}^{N\times 1}$, ${\bf{h}}_{ur}\in \mathbb{C}^{N\times 1}$ and ${\bf{h}}_{jr}\in \mathbb{C}^{N\times 1}$, respectively.
Besides, the channel vectors between zeRIS and AP/E are denoted by ${\bf{h}}_{ra}\in \mathbb{C}^{N\times 1}$ and ${\bf{h}}_{re}\in \mathbb{C}^{N\times 1}$, respectively.
Each entry of ${\bf{h}}_{pr}$, ${\bf{h}}_{ur}$, ${\bf{h}}_{jr}$, ${\bf{h}}_{ra}$ and ${\bf{h}}_{re}$, as well as $h_{pu}$ and $h_{pj}$ is reciprocal, and is assumed to be independently and identically complex Gaussian distribution with zero mean and unit variance.

During the ET stage $\tau T \ (0<\tau<1)$, PS transfers energy directly to U, J and zeRIS.
Then, by employing the linear EH model, the harvested energy of U, J and zeRIS can be respectively expressed as

\begin{align}
    \label{Eu}
    E_u = {\eta}{\tau}{T}{P_s} \left\lvert h_{pu} \right\rvert^{2}{\beta }_{pu},
\\
    \label{Ej}
    E_j = {\eta}{\tau}{T}{P_s} \left\lvert h_{pj} \right\rvert^{2}{\beta }_{pj},
\\
    \label{Eris}
    E_{ris} = {\eta}{\tau}{T}{P_s} \left\lvert\sum_{n = 1}^{N}{h}_{pr,n} \right\rvert^{2}{\beta }_{pr},
\end{align}
where $\eta$ denotes energy conversion efficiency that indicates the inherent inefficiencies and losses in the EH circuitry, $\tau$ denotes time allocation factor between ET and IT, and $P_s$ denotes transmission power at PS.
Besides, $\beta_{\kappa_{1}}= d_{\kappa_{1}}^{-a_{0}} (\kappa_1 \in \{ pu, pj, pr \})$ denotes the path loss from PS to U/J/zeRIS, where $d_{\kappa}$ denotes the corresponding distance between PS and U/J/zeRIS, and $a_0$ denotes the path loss exponent. 

To guarantee the normal operation of the zeRIS, it is crucial to determine its power consumption, which is caused by the configuration of reflecting elements, and the operation of the controller.
Therefore, considering that all reflecting elements are configured for information transmission at IT stage, the required energy for zeRIS is given by
\begin{equation}
    \label{Q}
    Q=(1-\tau)T(NP_{e}+P_{c}),
\end{equation}
where $P_e$ denotes the power consumption of each reflecting element, and $P_c$ denotes the power consumption of the zeRIS controller accountable for the setting of phase shift of each element.

During the IT stage $(1-\tau)T$, U exploits the harvested energy to transmit its confidential information to AP via zeRIS, while E would intercept this information which is nearby AP.
To address this problem, J is exploited to emit the artificial noise via zeRIS to combat against E.
The transmission power at U and J can be respectively written as
\begin{align}
    \label{Pu}
    {P_u} = \frac{E_u}{(1-\tau )T}  ={P_t}\left\lvert h_{pu} \right\rvert^{2}{\beta }_{pu},
\\
    \label{Pj}
    {P_j} = \frac{E_j}{(1-\tau )T}  ={P_t}\left\lvert h_{pj} \right\rvert^{2}{\beta }_{pj},
\end{align}  
where ${P_t}=\frac{ \tau }{1-\tau }P_s\eta $. 

Note that our work assumes that J can generate the artificial noise by a random Gaussian noise-like signal, whose codebook is pre-shared to AP but unknown to E. 
So that AP can subtract the artificial noise from its received signals, however, the eavesdropper has to be subject to the artificial noise \cite{9193903}. 
Hence, the received signals at AP and E can be expressed as
\begin{equation}
    \label{ya}
    {y_a}=\sqrt{P_u{\beta}_{ura}}{\bf{h}}_{ra}^{T}{\bf{\Theta }}{\bf{h}}_{{ur}}x_s + n_a,
\end{equation}
\begin{equation}
    \label{ye}
    {y_e}=\sqrt{P_u{\beta}_{ure}}{\bf{h}}_{re}^{T}{\bf{\Theta }}{\bf{h}}_{{ur}}x_s + \sqrt{P_j{\beta}_{jre}}{\bf{h}}_{re}^{T}{\bf{\Theta }}{\bf{h}}_{{jr}}x_j + n_e,
\end{equation}
where ${\bf{\Theta }} = {\rm diag} \{e^{j\phi_{1}},\cdots, e^{j\phi_{n}},\cdots , e^{j\phi_{N}}\} $ denotes $N \times N$ diagonal reflecting matrix of zeRIS, the amplitude of zeRIS is assumed to be equal to one in the case of full reflection, and $\phi_{n}$ denotes the phase shift of the $n$th reflecting element.
Moreover, $\beta_{ura}= (d_{ur}d_{ra})^{-a_{0}}$, $\beta_{ure}= (d_{ur}d_{re})^{-a_{0}}$ and $\beta_{jre}= (d_{jr}d_{re})^{-a_{0}}$ denote the path losses from U to AP/E and from J to E via zeRIS, respectively, where $d_{\kappa_{2}} (\kappa_2 \in \{ ur, jr, ra, re \})$ denotes the corresponding distance between zeRIS and U/J/AP/E.
$x_s$ denotes the intended signal from U to AP via zeRIS, while $x_j$ denotes the interference signal from J to E via zeRIS.
Besides, $n_a$ and $n_e$ denote the additive white Gaussian noise (AWGN) at AP and E, respectively.

As a result, the signal-to-noise ratio (SNR) of the received signal at AP and E can be respectively given by
\begin{align}
    \label{SNRa}
    {\rm{SNR}}_{a}={\rho }_t{\beta}_{pu}{\beta}_{ura}{\left\lvert h_{pu} \right\rvert }^{2}{\left\lvert {\bf{h}}_{ra}^{T}{\bf{\Theta }}{\bf{h}}_{{ur}} \right\rvert }^{2},
    \\
    \label{SNRe}
    {\rm{SNR}}_{e}=\frac{{\rho }_t{\beta}_{pu}{\beta}_{ure}{\left\lvert h_{pu} \right\rvert }^{2}{\left\lvert {\bf{h}}_{re}^{T}{\bf{\Theta }}{\bf{h}}_{{ur}} \right\rvert }^{2}}{{\rho }_t{\beta}_{pj}{\beta}_{jre}{\left\lvert h_{pj} \right\rvert }^{2}{\left\lvert {\bf{h}}_{re}^{T}{\bf{\Theta }}{\bf{h}}_{{jr}} \right\rvert }^{2} + 1},
\end{align}
where ${\rho }_t=\frac{P_t}{{\sigma}^2} $, and ${{\sigma}^2}$ is the variance of the AWGN.

\section{Proposed Modes and Metrics}
In this section, we first propose three new wireless-powered zeRIS modes. 
Then, three new metrics applicable to wireless-powered zeRIS are defined, which would be used for performance analysis of wireless-powered zeRIS communication systems in the next section.
\subsection{Proposed Modes}
In this subsection, we consider the inherent beamforming gain of zeRIS, which plays a pivotal role in enhancing the reliability and security of the system.
By configuring different numbers of reflecting elements at zeRIS, we propose three new wireless-powered zeRIS modes:
1) in mode-\Rmnum{1}, all (i.e., $N$) reconfigurable reflecting elements are adjusted to the optimal phase shift design of U to maximize the reliability of the system;
2) in mode-\Rmnum{2}, $N$ reconfigurable reflecting elements are adjusted to the optimal phase shift design of J to maximize the security of the system;
3) in mode-\Rmnum{3}, $N_1$ and $N_2$ $(N_1+N_2=N)$ reconfigurable reflecting elements are respectively adjusted to the optimal phase shift designs of U and J to balance the reliability and security of the system.
The specific phase shift designs on zeRIS in three modes are given as follows.

\subsubsection{Mode-\Rmnum{1}}
In mode-\Rmnum{1}, $N$ reconfigurable reflecting elements are adjusted to the optimal phase shift design of U.
We assume that U knows the artificial noise (pseudo-random noise) and perfectly cancels this.
Thus, ${\rm{SNR}}_{a}^{\Rmnum{1}}$ and ${\rm{SNR}}_{e}^{\Rmnum{1}}$ can be respectively rewritten as 
\begin{align}
    \label{SNRa1}
    {\rm{SNR}}_{a}^{\Rmnum{1}}={\rho }_t{\beta}_{pu}{\beta}_{ura}{\left\lvert h_{pu} \right\rvert }^{2}{\left\lvert \sum_{n = 1}^{N} h_{ur,n} h_{ra,n} e^{j\phi_{\Rmnum{1},n}}  \right\rvert }^{2},
\end{align}
\begin{align}
    \label{SNRe1}
    {\rm{SNR}}_{e}^{\Rmnum{1}}=\frac{{\rho }_t{\beta}_{pu}{\beta}_{ure}{\left\lvert h_{pu} \right\rvert }^{2}{\left\lvert \sum_{n = 1}^{N} h_{ur,n} h_{re,n} e^{j\phi_{\Rmnum{1},n}}  \right\rvert }^{2}}{{\rho }_t{\beta}_{pj}{\beta}_{jre}{\left\lvert h_{pj} \right\rvert }^{2}{\left\lvert \sum_{n = 1}^{N} h_{jr,n} h_{re,n} e^{j\phi_{\Rmnum{1},n}}  \right\rvert }^{2} + 1},
\end{align}
where the channel coefficient can be expressed in the product of the amplitude and the phase, thus, we have $h_{ur,n}=\left\lvert h_{ur,n} \right\rvert e^{j\phi_{h_{ur,n}}}$, and $h_{ra,n}=\left\lvert h_{ra,n} \right\rvert e^{j\phi_{h_{ra,n}}}$.
To maximize the information transmission of U, the optimal phase shifts $\phi_{\Rmnum{1},n}^{opt}$ are configured to match well with the phases of the cascaded channels from U to AP via zeRIS $h_{ur,n}$ and $h_{ra,n}$, i.e., $\phi_{\Rmnum{1},n}^{opt}=-(\phi_{h_{ur,n}}+\phi_{h_{ra,n}})$ \cite{10130095}.
Hence, we have
\begin{equation}
    \label{SNRa11}
    {\rm{SNR}}_{a}^{\Rmnum{1}}={\rho }_t{\beta}_{pu}{\beta}_{ura}{\left\lvert h_{pu} \right\rvert }^{2} \bigg ( \sum_{n = 1}^{N} \left\lvert h_{ur,n} \right\rvert \left\lvert h_{ra,n} \right\rvert   \bigg ) ^{2}.
\end{equation}

\subsubsection{Mode-\Rmnum{2}}
In mode-\Rmnum{2}, $N$ reconfigurable reflecting elements are adjusted to the optimal phase shift design of J.
${\rm{SNR}}_{a}^{\Rmnum{2}}$ and ${\rm{SNR}}_{e}^{\Rmnum{2}}$ can be respectively rewritten as 
\begin{align}
    \label{SNRa2}
    {\rm{SNR}}_{a}^{\Rmnum{2}}={\rho }_t{\beta}_{pu}{\beta}_{ura}{\left\lvert h_{pu} \right\rvert }^{2}{\left\lvert \sum_{n = 1}^{N} h_{ur,n} h_{ra,n} e^{j\phi_{\Rmnum{2},n}}  \right\rvert }^{2},
\end{align}
\begin{align}
    \label{SNRe2}
    {\rm{SNR}}_{e}^{\Rmnum{2}}=\frac{{\rho }_t{\beta}_{pu}{\beta}_{ure}{\left\lvert h_{pu} \right\rvert }^{2}{\left\lvert \sum_{n = 1}^{N} h_{ur,n} h_{re,n} e^{j\phi_{\Rmnum{2},n}}  \right\rvert }^{2}}{{\rho }_t{\beta}_{pj}{\beta}_{jre}{\left\lvert h_{pj} \right\rvert }^{2}{\left\lvert \sum_{n = 1}^{N} h_{jr,n} h_{re,n} e^{j\phi_{\Rmnum{2},n}}  \right\rvert }^{2} + 1},
\end{align}
where $h_{jr,n}=\left\lvert h_{jr,n} \right\rvert e^{j\phi_{h_{jr,n}}}$, and $h_{re,n}=\left\lvert h_{re,n} \right\rvert e^{j\phi_{h_{re,n}}}$.
To minimize the signal reception of E, the optimal phase shifts $\phi_{\Rmnum{2},n}^{opt}$ should be adjusted to perfectly match with the phases of the cascaded channels from J to E via zeRIS $h_{jr,n}$ and $h_{re,n}$, i.e., $\phi_{\Rmnum{2},n}^{opt}=-(\phi_{h_{jr,n}}+\phi_{h_{re,n}})$.
As such, ${\rm{SNR}}_{e}^{\Rmnum{2}}$ is yielded by
\begin{equation}
    \label{SNRe22}
    {\rm{SNR}}_{e}^{\Rmnum{2}}=\frac{{\rho }_t{\beta}_{pu}{\beta}_{ure}{\left\lvert h_{pu} \right\rvert }^{2}{\left\lvert {\bf{h}}_{re}^{T}{\bf{\Theta }}{\bf{h}}_{{ur}} \right\rvert }^{2}}{{\rho }_t{\beta}_{pj}{\beta}_{jre}{\left\lvert h_{pj} \right\rvert }^{2}\bigg ( \sum_{n = 1}^{N} \left\lvert h_{jr,n} \right\rvert \left\lvert h_{re,n} \right\rvert   \bigg ) ^{2} + 1}.
\end{equation}

\subsubsection{Mode-\Rmnum{3}}
In mode-\Rmnum{3}, $N_1$ and $N_2$ reconfigurable reflecting elements are respectively adjusted to the optimal phase shift designs of U and J.
To be specific, mode-\Rmnum{3} combines mode-\Rmnum{1} and mode-\Rmnum{2}, which develops an analysis of the segmented zeRIS.
The optimal phase shifts on $N_1$ elements $\phi_{\Rmnum{3},n_1}^{opt}$ are designed as $\phi_{\Rmnum{3},n_1}^{opt}=-(\phi_{h_{ur,n_1}}+\phi_{h_{ra,n_1}})$, while the optimal phase shifts on $N_2$ elements $\phi_{\Rmnum{3},n_2}^{opt}$ are designed as $\phi_{\Rmnum{3},n_2}^{opt}=-(\phi_{h_{jr,n_2}}+\phi_{h_{re,n_2}})$.
Consequently, ${\rm{SNR}}_{a}^{\Rmnum{3}}$ and ${\rm{SNR}}_{e}^{\Rmnum{3}}$ are given by
\begin{equation}
    \label{SNRa3}
    {\rm{SNR}}_{a}^{\Rmnum{3}}={\rho }_t{\beta}_{pu}{\beta}_{ura}{\left\lvert h_{pu} \right\rvert }^{2}\Delta_1,
\end{equation}
\begin{equation}
    \label{SNRe3}
    {\rm{SNR}}_{e}^{\Rmnum{3}}=\frac{{\rho }_t{\beta}_{pu}{\beta}_{ure}{\left\lvert h_{pu} \right\rvert }^{2}{\left\lvert {\bf{h}}_{re}^{T}{\bf{\Theta }}{\bf{h}}_{{ur}} \right\rvert }^{2}}{{\rho }_t{\beta}_{pj}{\beta}_{jre}{\left\lvert h_{pj} \right\rvert }^{2}\Delta_2  + 1},
\end{equation}
where
\begin{align}
    &\Delta_1 = {\left\lvert \sum_{n_1 = 1}^{N_1}  \left\lvert h_{ur,n_1} \right\rvert  \left\lvert h_{ra,n_1} \right\rvert + \hspace{-1em} \sum_{n_2 = N_1+1}^{N} \hspace{-1em} \left\lvert h_{ur,n_2} \right\rvert \left\lvert h_{ra,n_2} \right\rvert    e^{j\omega_{n_2}} \right\rvert }^{2}, \nonumber
    \\
    &\Delta_2 ={\left\lvert \sum_{n_1 = 1}^{N_1} \left\lvert h_{jr,n_1} \right\rvert  \left\lvert h_{re,n_1} \right\rvert   e^{j\omega_{n_1}} + \hspace{-1em} \sum_{n_2 = N_1+1}^{N} \hspace{-1em} \left\lvert h_{jr,n_2} \right\rvert  \left\lvert h_{re,n_2} \right\rvert   \right\rvert }^{2}, \nonumber
    \\
    &\omega_{n_1}=\phi_{h_{jr,n_1}}+\phi_{h_{re,n_1}} -(\phi_{h_{ur,n_1}}+\phi_{h_{ra,n_1}}) , \nonumber
    \\
    &\omega_{n_2}=\phi_{h_{ur,n_2}}+\phi_{h_{ra,n_2}} -(\phi_{h_{jr,n_2}}+\phi_{h_{re,n_2}}) . \nonumber
\end{align}
Since the channels satisfy the complex Gaussian distribution with zero mean and unit variance, the phases of the corresponding channels follow the uniform distribution $U(-\pi, \pi)$.
Moreover, the sum of the phase of two independent and identical channels still follows the uniform distribution $U(-\pi, \pi)$, and it is further deduced that $\omega_{n_1}$ and $\omega_{n_2}$ also follow the uniform distribution $U(-\pi, \pi)$, i.e., $\omega_{n_1}, \omega_{n_2}\thicksim U(-\pi, \pi)$.

\subsection{Proposed Metrics}
In this subsection, considering energy outage probability \cite{10348506}, data outage probability, and data intercept probability \cite{7152959}, we propose three new metrics for performance analysis applicable to wireless-powered zeRIS, called joint outage probability (JOP), joint intercept probability (JIP), and secrecy energy efficiency (SEE).

\subsubsection{Joint Outage Probability}
The JOP of the system is defined as the union of energy outage event and data outage event.
The energy outage event indicates that the zeRIS does not harvest a required amount of energy to maintain the zeRIS controller's stable operation.
Data outage event occurs when the main channel capacity is lower than a predefined data rate.
Then, JOP is mathematically expressed as
\begin{align}\label{JOPi1}
    P_{jop}^i={\rm Pr}[E_{ris} < Q \  \cup \  C_{a}^i<R],
\end{align}
where $E_{ris}$ and $Q$ are shown in (\ref{Eris}) and (\ref{Q}), $C_{a}^i=(1-\tau)\log_2(1+{\rm{SNR}}_{a}^i)$ and $R$ denote the main channel capacity and predefined data rate, respectively, and $i\in \{\Rmnum{1},\Rmnum{2},\Rmnum{3}\}$.

\subsubsection{Joint Intercept Probability}
The JIP of the system is defined as the intersection of energy sufficiency event and data interception event.
The energy sufficiency event indicates that the zeRIS harvests a sufficient amount of energy to maintain the zeRIS controller's stable operation.
Data interception event occurs when wiretap channel capacity is higher than a predefined data rate.
Then, JIP is denoted mathematically as
\begin{align}\label{JIPi1}
    P_{jip}^i={\rm Pr}[E_{ris} \geq  Q \  \cap  \  C_{e}^i \geq R],
\end{align}
where $E_{ris}$ and $Q$ are shown in (\ref{Eris}) and (\ref{Q}), and $C_{e}^i=(1-\tau)\log_2(1+{\rm{SNR}}_{e}^i)$ denotes the channel capacity achieved at the eavesdropper.

\subsubsection{Secrecy Energy Efficiency}
Based on the definition of energy efficiency in \cite{10348506}, we first define SEE in the context of secure wireless-powered zeRIS communication systems.
The SEE of the system is defined as the difference between legitimate energy efficiency and illegitimate energy efficiency.
Then, SEE is denoted mathematically as
\begin{align}\label{SEE}
    \varepsilon^{i}&=\{\varepsilon^{i}_{a}-\varepsilon^{i}_{e} \}^{+} \nonumber
    \\
    &=\frac{R}{P_s} \Big \{ 1- P^{i}_{jop}-P^{i}_{jip} \Big \}^{+},
\end{align}
where $\{x\}^{+}=\max\{x,0\}$.

\section{Performance Analysis}
In this section, we analyze the reliability, security and energy efficiency of three proposed modes for the wireless-powered zeRIS communication system.
To obtain further insights, we derive the closed-form expressions of JOP, JIP, and SEE for the proposed three modes, respectively.

Before launching into specific investigation, the cascaded channel statistics for segmented zeRIS are derived first to furnish a theoretical foundation for the subsequent performance analysis.
\subsection{Cascaded Channel Statistics for Segmented zeRIS}
Define $\xi_{n_1}=\left\lvert h_{ur,n_1} \right\rvert  \left\lvert h_{ra,n_1} \right\rvert$, $\xi_{n_2}=\left\lvert h_{ur,n_2} \right\rvert \left\lvert h_{ra,n_2} \right\rvert$, $\zeta_{n_1}=\left\lvert h_{jr,n_1} \right\rvert  \left\lvert h_{re,n_1} \right\rvert$, and $\zeta_{n_2}=\left\lvert h_{jr,n_2} \right\rvert \left\lvert h_{re,n_2} \right\rvert$.
Thus, $\Delta_{1}={\left\lvert \sum_{n_1 = 1}^{N_1}  \xi_{n_1} + \sum_{n_2 = N_1+1}^{N} \xi_{n_2}    e^{j\omega_{n_2}} \right\rvert }^{2}$, and $\Delta_{2}={\left\lvert \sum_{n_1 = 1}^{N_1} \zeta_{n_1}   e^{j\omega_{n_1}} + \sum_{n_2 = N_1+1}^{N} \zeta_{n_2}   \right\rvert }^{2}$.
Since the distinction between $\Delta_{1}$ and $\Delta_{2}$ is related to the reflecting numbers $N_1$ and $N_2$, the statistics of $\Delta_{1}$ and $\Delta_{2}$ are given in the following lemma.
\begin{lemma}
    \label{lemma1}
The first and second moments of the random variable (RV) $\Delta_{j}$ are denoted by $u_{\Delta_{j}}=\mathbb{E} \{ \Delta_{j}\}$ and $u_{\Delta_{j}}^{(2)}=\mathbb{E} \{ \Delta_{j}^{2}\}$, respectively, where $j \in\{1,2\}$.
The distribution of $\Delta_{j}$ can be regarded as a Gamma distribution, whose shape parameter $k_{j}$ and scale parameter $\theta_{j}$ can be given by \cite{9599656}
\begin{equation}
    \label{k1theta1}
    k_{j}=\frac{u_{\Delta_{j}}^2}{u_{\Delta _{j}}^{(2)}-u_{\Delta_{j}}^2}, \ \theta_{j}=\frac{u_{\Delta_{j}}^{(2)}-u_{\Delta_{j}}^2}{u_{\Delta_{j}}}.
\end{equation}
The PDF and CDF of $\Delta_{j}$ can be respectively expressed as
\begin{align}
    \label{PDF1}
    f_{\Delta_{j}}(z)&=\frac{z^{k_{j}-1}}{\Gamma(k_j)\theta_{j}^{k_j}}e^{-\frac{z}{\theta_{j}} }, 
    \\
    \label{CDF1}
    F_{\Delta_{j}}(z)&=\frac{1}{\Gamma(k_j)}\gamma \bigg(k_{j},\frac{z}{\theta_{j}} \bigg).
\end{align}
\end{lemma}
\begin{proof}
    See Appendix \ref{AppendixA}.
\end{proof}

\subsection{JOP Analysis}
It can be observed from (\ref{Eris}), (\ref{SNRa11}) and (\ref{SNRa2}) that the random variables (RVs) in energy outage event and data outage event are independent of each other.
Hence, the definition of JOP in (\ref{JOPi1}) can be rewritten as
\begin{align}\label{JOPi2}
    P_{jop}^i=A+B^{i}-AB^{i},   
\end{align}
where $A={\rm Pr}[E_{ris} < Q]$ and $B^{i}={\rm Pr}[C_{a}^i<R]$ denote the probability of energy outage event and data outage event, respectively.

\subsubsection{Accurate Analysis}
\begin{theorem}
    \label{theorem1}
    The JOP of mode-\Rmnum{1} is denoted as
    \begin{align}
        \label{JOP1}
        P_{jop}^{\Rmnum{1}}=&1- \frac{\pi^2}{4L\Gamma(v)\varphi^v} e^{-\frac{NP_{e}+P_{c}}{P_t N\beta_{pr}}} \sum_{l = 1}^{L} \sqrt{1-\varpi_l^2}\sec^2u_l  \nonumber
        \\
        &\times  (\tan u_l)^{v-1} \exp \bigg( -\frac{\varsigma}{\rho_t \tan^2u_l}-\frac{\tan u_l}{\varphi}  \bigg),
    \end{align}
where $\epsilon=2^{\frac{R}{1-\tau} }-1 $, $u_l=\frac{(\varpi  _l+1)}{4}\pi $, $\varpi_l=\cos(\frac{2l-1}{2L}\pi )$, $\varsigma =\frac{\epsilon }{\beta_{pu} \beta_{ura}}$, and $L$ is the accuracy versus complexity parameter.
\end{theorem}

\begin{proof}
    See Appendix \ref{AppendixB}-\Rmnum{1}.
\end{proof}

\begin{theorem}
    \label{theorem2}
    The JOP of mode-\Rmnum{2} is denoted as
    \begin{align}
        \label{JOP2}
        P_{jop}^{\Rmnum{2}}=1 - e^{-\frac{NP_{e}+P_{c}}{P_t N\beta_{pr}}} \sqrt{\frac{4\varsigma}{\rho_t N } } K_1 \Bigg(\sqrt{\frac{4\varsigma}{\rho_t N } } \Bigg),
\end{align}
where $K_1(\cdot)$ is the first order modified Bessel function of the second kind \cite[Eq. (8.432)]{10.1115/1.3138251}.
\end{theorem}

\begin{proof}
    See Appendix \ref{AppendixB}-\Rmnum{2}.
\end{proof}

\begin{theorem}
    \label{theorem3}
    The JOP of mode-\Rmnum{3} is denoted as
    \begin{align}
        \label{JOP3}
        P_{jop}^{\Rmnum{3}} =& 1 - \frac{2}{\Gamma(k_1)} e^{-\frac{NP_{e}+P_{c}}{P_t N\beta_{pr}}} \bigg(\frac{\varsigma}{\rho_t \theta_1} \bigg)^{\frac{k_1}{2} } K_{k_1}\bigg(2\sqrt{\frac{\varsigma}{\rho_t \theta_1}}\bigg),
\end{align}
where $K_{i}(\cdot)$ is the $i$th order modified Bessel function of the second kind \cite[Eq. (8.432)]{10.1115/1.3138251}.
\end{theorem}

\begin{proof}
    See Appendix \ref{AppendixB}-\Rmnum{3}.
\end{proof}

\begin{remark}
    \label{remark1}
    It is observed from (\ref{B1}), (\ref{B2}) and (\ref{B3}) that the probability of data outage event in mode-\Rmnum{3} is lower than that in mode-\Rmnum{2} but larger than that in mode-\Rmnum{1}, such that mode-\Rmnum{1} achieves the best JOP, while mode-\Rmnum{2} achieves the worst JOP.
    The reason behind this conclusion is that the energy transfer process is the same in three modes, and the number of zeRIS elements that adopt an optimal phase shift design with respect to U is $N$, 0, and $N_1$ in the three modes, respectively.
    Hence, among the three modes, the data transmission capability of mode-\Rmnum{1} is the strongest, while that of mode-\Rmnum{2} is the worst, and that of mode-\Rmnum{3} strikes a reliability trade-off between mode-\Rmnum{1} and mode-\Rmnum{2}.
    Besides, this enhanced effect can be further improved by increasing the number of zeRIS elements.
\end{remark}

\subsubsection{Asymptotic Analysis}
\begin{proposition}
    \label{proposition1}
    The asymptotic JOP of mode-\Rmnum{1} can be given by
    \begin{align}
        \label{JOP1_asy}
        P_{jop,asy}^{\Rmnum{1}} =&  \frac{1}{P_t} \bigg( \frac{NP_{e}+P_{c}}{ N\beta_{pr}} \nonumber
        \\
        &+ \frac{ \big(2^{\frac{R}{1-\tau} }-1 \big)\sigma^2  }{\beta_{pu} \beta_{ura} \varphi^2 \big(\frac{N\pi^2}{16-\pi^2}-1 \big )\big(\frac{N\pi^2}{16-\pi^2}-2\big)} \bigg),
    \end{align}
    where ${P_t}=\frac{ \tau }{1-\tau }P_s\eta $. 
\end{proposition}

\begin{proof}
    When the transmission power is large, i.e., $P_s \to \infty$, by invoking the first-order Taylor expansion $\exp(-x)\to 1-x \ (x \to 0)$ into (\ref{A}) and (\ref{B1}), and applying \cite[Eq. (3.381.4)]{10.1115/1.3138251} in the case of $v>2$, we have
    \begin{align}
        \label{A_asy}
    A_{asy}=&\frac{NP_{e}+P_{c}}{P_t N\beta_{pr}},
    \end{align}
    \begin{align}
        \label{B1_asy}
    B^{\Rmnum{1}}_{asy}=\frac{\big(2^{\frac{R}{1-\tau} }-1 \big)\sigma^2   }{\rho_t \beta_{pu} \beta_{ura} \varphi^2 \big(\frac{N\pi^2}{16-\pi^2}-1 \big )\big(\frac{N\pi^2}{16-\pi^2}-2\big)}.
    \end{align}
    Substituting (\ref{A_asy}) and (\ref{B1_asy}) into (\ref{JOPi2}), the result of (\ref{JOP1_asy}) is obtained.
\end{proof}

\begin{proposition}
    \label{proposition2}
    The asymptotic JOP of mode-\Rmnum{2} can be given by
    \begin{align}
        \label{JOP2_asy}
        P_{jop,asy}^{\Rmnum{2}} =  \frac{\big(2^{\frac{R}{1-\tau} }-1 \big)\sigma^2  }{P_t \beta_{pu} \beta_{ura} N} \ln \Bigg(\frac{P_t \beta_{pu} \beta_{ura} N}{\big(2^{\frac{R}{1-\tau} }-1 \big)\sigma^2 }\Bigg).
    \end{align}
\end{proposition}

\begin{proof}
    When the transmission power is large, i.e., $P_s \to \infty$, by invoking the approximation $\vartheta K_1 (\vartheta) \to 1 + \frac{\vartheta^2}{2} \ln(\frac{\vartheta}{2}) \ (\vartheta \to 0)$ into (\ref{B2}), we have
    \begin{align}
        \label{B2_asy}
    B^{\Rmnum{2}}_{asy}=\frac{\big(2^{\frac{R}{1-\tau} }-1 \big)\sigma^2  }{P_t \beta_{pu} \beta_{ura} N} \ln \Bigg(\frac{P_t \beta_{pu} \beta_{ura} N}{\big(2^{\frac{R}{1-\tau} }-1 \big)\sigma^2 }\Bigg).
    \end{align}
    Substituting (\ref{A_asy}) and (\ref{B2_asy}) into (\ref{JOPi2}), we find that when $P_s \to \infty$, $A_{asy}$ is much smaller than $B^{\Rmnum{2}}_{asy}$.
    Hence, the terms of $A_{asy}$ and $A_{asy}B^{\Rmnum{2}}_{asy}$ can be omitted, and the result of (\ref{JOP2_asy}) is obtained.
\end{proof}

\begin{proposition}
    \label{proposition3}
    The asymptotic JOP of mode-\Rmnum{3} can be given by
    \begin{align}
        \label{JOP3_asy}
        P_{jop,asy}^{\Rmnum{3}} =  \frac{1}{P_t} \bigg( \frac{NP_{e}+P_{c}}{ N\beta_{pr}} + \frac{\big(2^{\frac{R}{1-\tau} }-1 \big)\sigma^2  }{\beta_{pu} \beta_{ura}\theta_1 (k_1 - 1)} \bigg).
    \end{align}
\end{proposition}

\begin{proof}
    When the transmission power is large, i.e., $P_s \to \infty$, by invoking the first-order Taylor expansion $\exp(-x)\to 1-x \ (x \to 0)$ into the penultimate line of (\ref{B3}), and applying \cite[Eq. (3.381.4)]{10.1115/1.3138251} in the case of $k_1>1$, we have
    \begin{align}
        \label{B3_asy}
    B^{\Rmnum{3}}_{asy}=\frac{\big(2^{\frac{R}{1-\tau} }-1 \big)\sigma^2   }{P_t \beta_{pu} \beta_{ura} \theta_1 (k_1 - 1)}.
    \end{align}
    Similar to the proof of Proposition \ref{proposition1}, the result of (\ref{JOP3_asy}) is obtained.
\end{proof}

\begin{remark}
    \label{remark2}
    Proposition \ref{proposition1} and Proposition \ref{proposition3} show that the asymptotic JOP is the sum of the probabilities of energy outage event (i.e., the first term of the sum) and data outage event (i.e., the second term of the sum), which is different from outage probability performance of conventional RIS.
    Moreover, as the transmission power increases, the impact of small-scale fading caused by the randomness of the channel is eliminated, and then the two events are treated as independent.
\end{remark}

\begin{remark}
    \label{remark3}
    Proposition \ref{proposition1}, Proposition \ref{proposition2}, and Proposition \ref{proposition3} show that all the diversity gains of three modes are 1, since the channels between PS and zeRIS and the channel between PS and U/J follow the uniform phase distribution, so that the energy signal is randomly reflected in space, and the harvested energy of U/J/zeRIS scales equally with transmission power.
    In addition, the asymptotic JOP is further improved by increasing the number of zeRIS elements $N$.
\end{remark}

\begin{corollary}
    \label{corollary1}
    When $P_s \to \infty$ and $N \to \infty$, the asymptotic JOP of mode-\Rmnum{1} is further given by
    \begin{align}
        \label{JOP1_asyN}
        P_{jop,asy}^{\Rmnum{1}} =  \frac{P_{e}}{ \beta_{pr}} P_t^{-1}.
    \end{align}
\end{corollary}

\begin{proof}
    It is observed from (\ref{A_asy}) and (\ref{B1_asy}) that when $N \to \infty$, $A_{asy}$ tends to a constant, and $B^{\Rmnum{1}}_{asy}$ continues to decrease to a negligible degree.
    Then, the result of (\ref{JOP1_asyN}) is obtained.
\end{proof}

\begin{corollary}
    \label{corollary2}
    When $P_s \to \infty$ and $N \to \infty$, the asymptotic JOP of mode-\Rmnum{2} is further given by
    \begin{align}
        \label{JOP2_asyN}
        P_{jop,asy}^{\Rmnum{2}} =  \frac{P_{e}}{ \beta_{pr}} P_t^{-1}.
    \end{align}
\end{corollary}

\begin{proof}
    It is observed from (\ref{B2_asy}) that when $P_s \to \infty$ and $N \to \infty$, $B^{\Rmnum{2}}_{asy}$ tends to zero, and is much smaller than $A_{asy}$.
    Hence, the result of (\ref{JOP2_asyN}) is obtained.
\end{proof}

\begin{corollary}
    \label{corollary3}
    When $P_s \to \infty$ and $N \to \infty$, the asymptotic JOP of mode-\Rmnum{3} is further given by
    \begin{align}
        \label{JOP3_asyN}
        P_{jop,asy}^{\Rmnum{3}} =  \frac{P_{e}}{ \beta_{pr}} P_t^{-1}.
    \end{align}
\end{corollary}

\begin{proof}
    It is observed from (\ref{k1theta1}), (\ref{EDelta1_final}) and (\ref{EDelta12_final}) that when $N \to \infty$, $\theta_1 (k_1 - 1)\approx \theta_1k_1 \approx \frac{\pi^2}{16}N_1^2$.
    Similar to the proof of Corollary \ref{corollary1}, the result of (\ref{JOP3_asyN}) is obtained.
\end{proof}

\begin{remark}
    \label{remark4}
    Corollary \ref{corollary1}, Corollary \ref{corollary2}, and Corollary \ref{corollary3} show that by increasing the number of zeRIS elements $N$, the asymptotic JOP for all three modes approaches the same performance ceiling, which is determined by the energy outage event.
    This is because under high transmission power $P_s$ and large $N$, data outage event hardly occurs, while the probability of energy outage event reaches a floor as $N$ increases.
\end{remark}

\begin{remark}
    \label{remark5}
    As the increases of transmission power $P_s$ and the number of zeRIS elements $N$, the JOP mainly depends on the power consumption of each reflecting element and the path loss of channel between PS and zeRIS.
    Particularly, the energy consumption and energy harvesting of zeRIS increase in parallel as $N$ increases, and the energy consumption proportion of the zeRIS controller gradually decreases to zero.
    Moreover, the path loss of channel between PS and zeRIS becomes a bottleneck in the effects of all channels as the increases of $P_s$ and $N$, which is due to the constraint of energy outage at the zeRIS.  
\end{remark}

\subsection{JIP Analysis}
Similarly, it is observed from (\ref{Eris}), (\ref{SNRe1}), and (\ref{SNRe22}) that the RVs in energy sufficiency event and data interception event are independent of each other.
Therefore, the definition of JIP in (\ref{JIPi1}) can be rewritten as
\begin{align}\label{JIPi2}
    P_{jip}^i=CD^{i},
\end{align}
where $C={\rm Pr}[E_{ris} \geq Q]$ and $D^{i}={\rm Pr}[C_{e}^i \geq R]$ denote the probability of energy sufficiency event and data interception event, respectively.

\subsubsection{Accurate Analysis}
\begin{theorem}
    \label{theorem4}
    The JIP of mode-\Rmnum{1} is given by
    \begin{align}
        \label{JIP1}
        P_{jip}^{\Rmnum{1}}=&-  \frac{\pi^2 \beta_{pu}}{4L \epsilon N^2 \beta_{pj} \beta_{ure} \beta_{jre}}  e^{-\frac{NP_{e}+P_{c}}{P_t N\beta_{pr}}} \sum_{l = 1}^{L} \sqrt{1-\varpi_l^2} \nonumber
        \\
        &\times \sec^2u_l \tan u_l \exp \bigg( \frac{\beta_{pu}\beta_{ure} - \epsilon \beta_{pj}\beta_{jre}}{\epsilon N \beta_{pj} \beta_{jre} \beta_{ure}}\tan u_l  \nonumber
        \\
        & \hspace{1em} - \frac{\epsilon}{\rho_t \beta_{pu} \tan u_l}  \bigg) \times {\rm Ei} \bigg( -\frac{\beta_{pu} \tan u_l}{\epsilon N \beta_{pj} \beta_{jre}}  \bigg),
    \end{align}
where $\rm Ei(\cdot)$ is the exponential integral function \cite[Eq. (8.211)]{10.1115/1.3138251}.
\end{theorem}

\begin{proof}
    See Appendix \ref{AppendixC}-\Rmnum{1}.
\end{proof}

\begin{theorem}
    \label{theorem5}
    The JIP of mode-\Rmnum{2} is given by
    \begin{align}
        \label{JIP2}
        P_{jip}^{\Rmnum{2}}=&\frac{\pi^2}{8L N \beta_{ure} (v-1) \varphi^v} \bigg(\frac{\beta_{pu}}{\epsilon \beta_{pj} \beta_{jre}} \bigg)^{\frac{v}{2}} e^{-\frac{NP_{e}+P_{c}}{P_t N\beta_{pr}}}  \nonumber
        \\
        &\times \sum_{l = 1}^{L} \sqrt{1-\varpi_l^2} \sec^2u_l (\tan u_l)^{\frac{v}{2}} e^{-\frac{\epsilon}{\rho_t \beta_{pu} \tan u_l}-\frac{\tan u_l}{N \beta_{ure}}} \nonumber
        \\
        &\times \bigg[  \exp\left (\alpha_1  +i\frac{(v-2)\pi}{2}  \right ) \Gamma\Big(2-v, \alpha_1 \Big) \nonumber
        \\
        & + \exp\left (-\alpha_1 -i\frac{(v-2)\pi}{2}  \right )  \Gamma\Big(2-v, -\alpha_1 \Big) \bigg],
    \end{align}
where $\alpha_1 =\frac{i}{\varphi}\sqrt{\frac{\beta_{pu}\tan u_l}{\epsilon\beta_{pj}\beta_{jre}} }$, and $\Gamma(\cdot,\cdot)$ is the upper incomplete gamma function \cite[Eq. (8.350.2)]{10.1115/1.3138251}.
\end{theorem}

\begin{proof}
    See Appendix \ref{AppendixC}-\Rmnum{2}.
\end{proof}

\begin{theorem}
    \label{theorem6}
    The JIP of mode-\Rmnum{3} is given by
    \begin{align}
        \label{JIP3}
        P_{jip}^{\Rmnum{3}}=&\frac{\pi^2}{4L N \beta_{ure} \theta_2^{k_2}} \bigg(\frac{\beta_{pu}}{\epsilon \beta_{pj} \beta_{jre}} \bigg)^{k_2} e^{-\frac{NP_{e}+P_{c}}{P_t N\beta_{pr}}}  \nonumber
        \\
        &\times \sum_{l = 1}^{L} \sqrt{1-\varpi_l^2} \sec^2u_l  \Gamma \bigg (1-k_2,\frac{\beta_{pu} \tan u_l}{\epsilon \beta_{pj} \beta_{jre} \theta_2 } \bigg) \nonumber
        \\
        &\times (\tan u_l)^{k_2} \exp \bigg[ \bigg( \frac{\beta_{pu}}{\epsilon \beta_{pj} \beta_{jre} \theta_2} - \frac{1}{N\beta_{ure}} \bigg) \tan u_l \nonumber
        \\
        &\hspace{1em} - \frac{\epsilon}{\rho_t \beta_{pu} \tan u_l} \bigg].
    \end{align}
\end{theorem}

\begin{proof}
    See Appendix \ref{AppendixC}-\Rmnum{3}.
\end{proof}

\begin{remark}
    \label{remark6}
    It is concluded from (\ref{D1}), (\ref{D2}) and (\ref{D3}) that the probability of data interception event in mode-\Rmnum{3} is higher than that in mode-\Rmnum{2} but lower than that in mode-\Rmnum{1}, such that mode-\Rmnum{1} achieves the worst JIP, while mode-\Rmnum{2} achieves the best JIP, and mode-\Rmnum{3} achieves a security trade-off between mode-\Rmnum{1} and mode-\Rmnum{2}.
    The reason behind the results is that the number of zeRIS elements that adopt an optimal phase shift design with respect to cooperative jamming user is 0, $N$, and $N_2$ in the three modes, respectively, thus mode-\Rmnum{2} achieves the most robust secrecy transmission capability among the three modes.
\end{remark}

\subsubsection{Asymptotic Analysis}
\begin{proposition}
    \label{proposition4}
    The asymptotic expression of $P_{jip}^{\Rmnum{1}}$ is expressed as 
    \begin{align}
        \label{JIP1_asy}
        P_{jip,asy}^{\Rmnum{1}}= -\mu_1 (\ln \mu_1 +\psi (2)),
    \end{align}
    where $\mu_1=\frac{\beta_{pu}\beta_{ure}}{ \big(2^{\frac{R}{1-\tau} }-1 \big) \beta_{pj}\beta_{jre}} $, and $\psi(\cdot)$ is the psi function \cite[Eq. (8.36)]{10.1115/1.3138251}.
\end{proposition}

\begin{proof}
    Based on (\ref{C}) and (\ref{D1}), when $P_s \to \infty$, we find that $C=1$, and the term in $D^{\Rmnum{1}}$ whose denominators includes $\rho_t$ can be omitted.
    Besides, according to \cite[Eq. (8.212.1)]{10.1115/1.3138251} and \cite[Eq. (8.367.4)]{10.1115/1.3138251}, we have ${\rm Ei}(-x) \to e^{-x}\ln x $ when $x \to \infty$.
    Then, by using \cite[Eq. (4.352.1)]{10.1115/1.3138251}, Eq. (\ref{JIP1_asy}) is obtained.
    Thus, the proof of Proposition \ref{proposition4} is completed.
\end{proof}

\begin{proposition}
    \label{proposition5}
    The asymptotic expression of $P_{jip}^{\Rmnum{2}}$ is expressed as 
    \begin{align}
        \label{JIP2_asy}
        P_{jip,asy}^{\Rmnum{2}}=&1+\frac{\mu_2 \pi^2}{4L\Gamma(v)\varphi^v} \sum_{l = 1}^{L} \sqrt{1-\varpi_l^2} \sec^2u_l (\tan u_l)^{v+1} \nonumber
        \\
        &\times e^{\mu_2 \tan^2u_l-\frac{\tan u_l}{\varphi} } {\rm Ei} \big( -\mu_2 \tan^2u_l \big),
    \end{align}
    where $\mu_2=\frac{1}{N\mu_1}$.
\end{proposition}

\begin{proof}
    Following similar proof steps with Proposition \ref{proposition4}, we omit the term that contains $\rho_t$ in the denominator, and exchange the order of integration of $y_1$ and $y_2$ in (\ref{D2}). 
    By utilizing \cite[Eq. (3.352.4)]{10.1115/1.3138251} and Gaussian-Chebyshev quadrature, the result of (\ref{JIP2_asy}) is derived.
\end{proof}

\begin{proposition}
    \label{proposition6}
    The asymptotic expression of $P_{jip}^{\Rmnum{3}}$ is expressed as 
    \begin{align}
        \label{JIP3_asy}
        P_{jip,asy}^{\Rmnum{3}}=&1+\frac{\mu_2 \pi^2}{4L\Gamma(k_2)\theta_2^{k_2}} \sum_{l = 1}^{L} \sqrt{1-\varpi_l^2} \sec^2u_l (\tan u_l)^{k_2} \nonumber
        \\
        &\times e^{ \big( \mu_2 -\frac{1}{\theta_2} \big) \tan u_l } {\rm Ei}\big(-\mu_2 \tan u_l \big).
    \end{align}
\end{proposition}

\begin{proof}
    Similar to the proof of Proposition \ref{proposition5}, the result in (\ref{JIP3_asy}) is obtained.
    We skip the proof for brevity.
\end{proof}

\begin{remark}
    \label{remark7}
    It can be observed from Proposition \ref{proposition4}, Proposition \ref{proposition5}, and Proposition \ref{proposition6} that as transmission power $P_s$ increases, the JIP will converge to a performance floor.
    It is evident from (\ref{JIP1_asy}) that this performance floor in mode-\Rmnum{1} is independent of the number of zeRIS elements $N$, this is due to the fact that the terms $Y_{\Rmnum{1},1}$ and $Y_{\Rmnum{1},2}$ in (\ref{D1}) follow the exponential distribution whose parameters are $N\beta_{ure}$ and $N\beta_{jre}$, respectively, and then $N$ can be omitted.
    In addition, the performance floors in mode-\Rmnum{2} and mode-\Rmnum{3} can be enhanced by increasing $N$.
\end{remark}

To further obtain the asymptotic expressions of JIP in the case of high transmission power and large number of zeRIS element, the LoS links between PS and U/J are assumed to establish, i.e., $\left\lvert h_{pu}\right\rvert = \left\lvert h_{pj}\right\rvert =1$.

\begin{corollary}
    \label{corollary4}
    When $P_s \to \infty$ and $N \to \infty$, the asymptotic JIP of mode-\Rmnum{2} is further expressed as
    \begin{align}
        \label{JIP2_asyN}
        P_{jip,asy}^{\Rmnum{2}} =  \bigg(\frac{e \mu_3}{2v} \bigg)^{\frac{v}{2} } v^{-\frac{1}{4} } \exp \bigg(\frac{\mu_3}{8} -\sqrt{\frac{v\mu_3}{2} }  \bigg),
    \end{align}
where $\mu_3=\frac{N\mu_1}{\varphi^2}$, and the definition of $\mu_1$ can be found in (\ref{JIP1_asy}).
\end{corollary}

\begin{proof}
    When $P_s \to \infty$ and $N \to \infty$, using \cite[Eqs. (3.462.1), (9.240), (9.229.2)]{10.1115/1.3138251} and the expression $e^{-\frac{v}{2}\ln\frac{v}{2} } = (\frac{v}{2})^{-\frac{v}{2}}$, the result of (\ref{JIP2_asyN}) is obtained.
\end{proof}

\begin{corollary}
    \label{corollary5}
    When $P_s \to \infty$ and $N \to \infty$, the asymptotic JIP of mode-\Rmnum{3} is further expressed as
    \begin{align}
        \label{JIP3_asyN}
        P_{jip,asy}^{\Rmnum{3}} =  \bigg( \frac{N\mu_1}{\theta_2 + N\mu_1} \bigg )^{k_2}.
    \end{align}
\end{corollary}

\begin{proof}
    When $P_s \to \infty$ and $N \to \infty$, applying \cite[Eq. (3.381.4)]{10.1115/1.3138251}, the result of (\ref{JIP3_asyN}) is obtained.
\end{proof}

\begin{remark}
    \label{remark8}
    It can be observed that in Corollary \ref{corollary4}, $v=\frac{N\pi^2}{16-\pi^2} \gg \mu_3$, and in Corollary \ref{corollary5}, $k_2$ is related to $N_2$, such that the asymptotic JIP of mode-\Rmnum{2} and mode-\Rmnum{3} can be further improved by increasing the number of zeRIS elements $N$.
    This result can be explained by the fact that the number of zeRIS elements that adopting the optimal phase shift with respect to cooperative jamming user in mode-\Rmnum{2} and mode-\Rmnum{3} is $N$ and $N_2$, respectively.
    Under high transmission power, as the number of zeRIS elements increases, the artificial noise at eavesdropper gets stronger than the information intercepted by eavesdropper, making the SNR of eavesdropper worse, thus improving the security of the system.
\end{remark}

\begin{remark}
    \label{remark9}
    Corollary \ref{corollary4} and Corollary \ref{corollary5} show that the asymptotic JIP of mode-\Rmnum{2} and mode-\Rmnum{3} is independent of the channel between PS and zeRIS, since the energy sufficiency event is almost certain to occur under high transmission power.
    In addition, the asymptotic JIP performance floor can be further improved by correctly designing the time switching factor, predefined data rate, and the distance between the corresponding nodes.
\end{remark}

\subsection{SEE Analysis}
Substituting (\ref{JOP1}), (\ref{JOP2}), (\ref{JOP3}), (\ref{JIP1}), (\ref{JIP2}) and (\ref{JIP3}) into (\ref{SEE}), the SEE in (\ref{SEE}) of three modes is derived.

\begin{remark}
\label{remark10}
Due to the inclusion of JOP, JIP, predefined data rate and transmission power in the SEE definition, we can quantitatively evaluate the performance of wireless-powered zeRIS communication system in terms of security, reliability and energy consumption.
\end{remark}

\section{Simulation Results and Discussion}
In this section, we employ Monte Carlo simulations to validate the theoretical analysis.
The concept of security-reliability trade-off (SRT) refers to treating JOP as a function of JIP, and SRT performance is also reflected in normalized joint intercept and outage probability (JIOP), which is the mean of the sum of JOP and JIP.
To validate the performance improvement of the friendly cooperative jamming user J in terms of system security, a wireless-powered zeRIS communication system without the aid of J is considered as benchmark-\Rmnum{1}.
Moreover, to demonstrate the superiority of proposed three modes in terms of security, reliability and energy efficiency, a wireless-powered zeRIS communication system without optimal phase shift design is considered as benchmark-\Rmnum{2}.
Unless otherwise stated, we assume the distances from zeRIS to PS, U, J, AP and E, and the distances from PS to U and J are $10$m, respectively.
Moreover, the other simulation parameters are set as $a_0 =2.7$, $T=1$, $\tau=0.4$, $\eta =0.8$, $R=1.5$ BPCU, $N=30$, $N_1=N_2=\frac{N}{2}=15$, $P_e=2$ $\mu$W, $P_c=50$ mW, $\sigma ^2=-45$ dBm, and $L=1500$.

Fig. \ref{fig2} shows that the JOP versus transmission power $P_s$ for proposed modes.
It is observed from the figure that the accurate and asymptotic closed-form expressions of JOP match well with the simulation results, verifying the correctness of theoretical analysis of JOP.
Moreover, the figure shows that under the same condition, mode-\Rmnum{1} achieves the best JOP, while mode-\Rmnum{2} achieves the worst JOP, and mode-\Rmnum{3} is in the middle of this comparison, which confirms the conclusion in \textit{Remark \ref{remark1}}.
This is attributed to the fact that the number of zeRIS elements that adopt an optimal phase shift design with respect to U is $N$, 0, and $N_1$ in the three modes, respectively, which enhances the information transmission quality and the reliability of the system in various degrees.
In addition, it is clear that the asymptotic JOP lines are straight lines with the slope of 1, which means that all the diversity gains of three modes are 1, confirming the conclusion in \textit{Remark \ref{remark3}}.
As can be seen from the figure, the increase in energy transfer distances $d_{pr},d_{pu}$ or information transmission distances $d_{ur},d_{ar}$ leads to an increase in JOP.
Differently, in mode-\Rmnum{1}, the increased $d_{pr}$ and $d_{pu}$ have a stronger deterioration effect than the increased $d_{ur}$ and $d_{ar}$, while the conclusion is opposite in mode-\Rmnum{2} and mode-\Rmnum{3}.
It can be interpreted that the fading effect caused by increasing distance on information transmission over the zeRIS cascaded channels in mode-\Rmnum{1} is weaker than that in mode-\Rmnum{2} and mode-\Rmnum{3}.

\begin{figure}[t]
    \centering
    \includegraphics[width=0.6\linewidth]{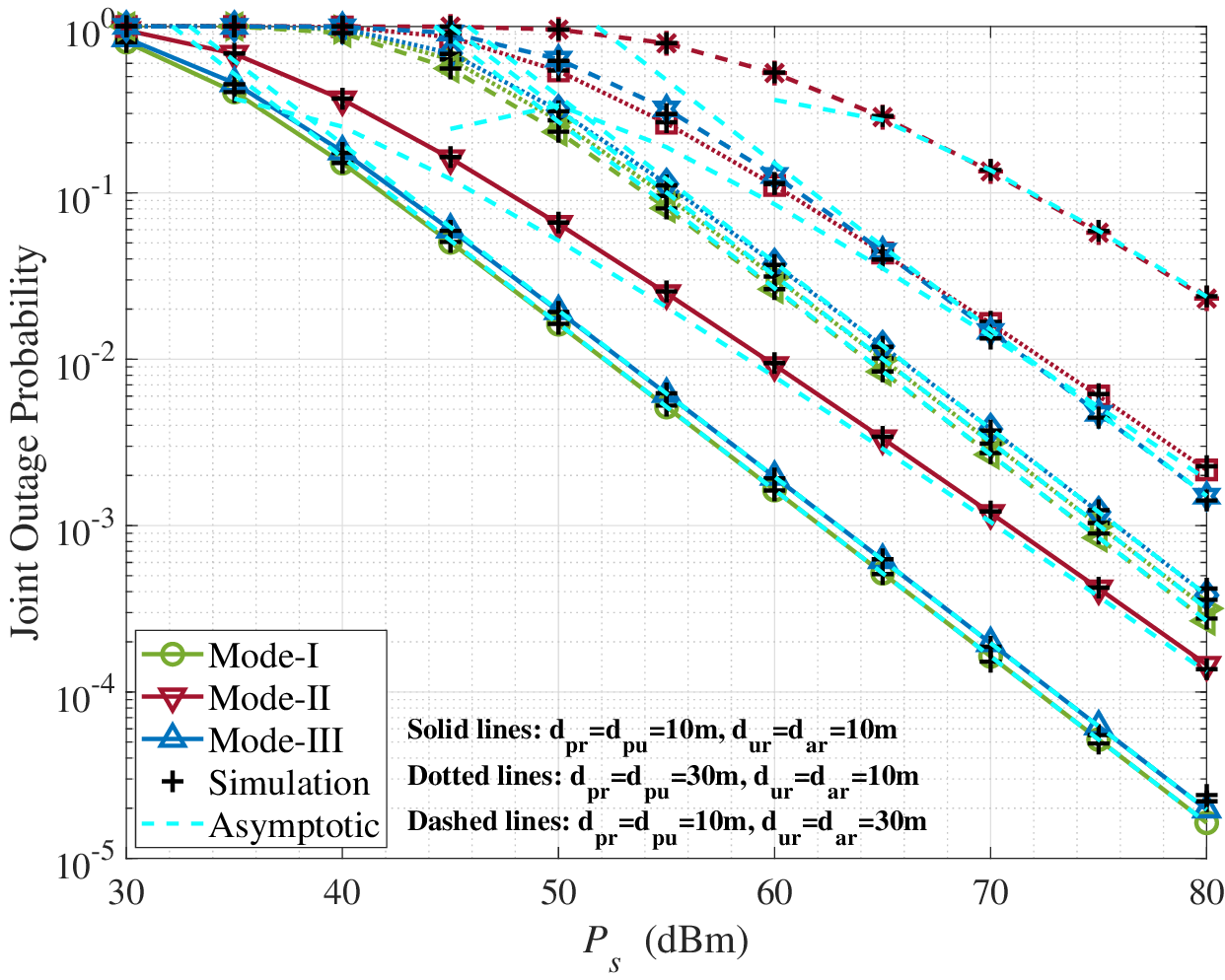}
    \caption{The JOP versus $P_s$ for proposed modes, where $N=30$.}
    \label{fig2}
\end{figure}

\begin{figure}[t]
    \centering
    \includegraphics[width=0.6\linewidth]{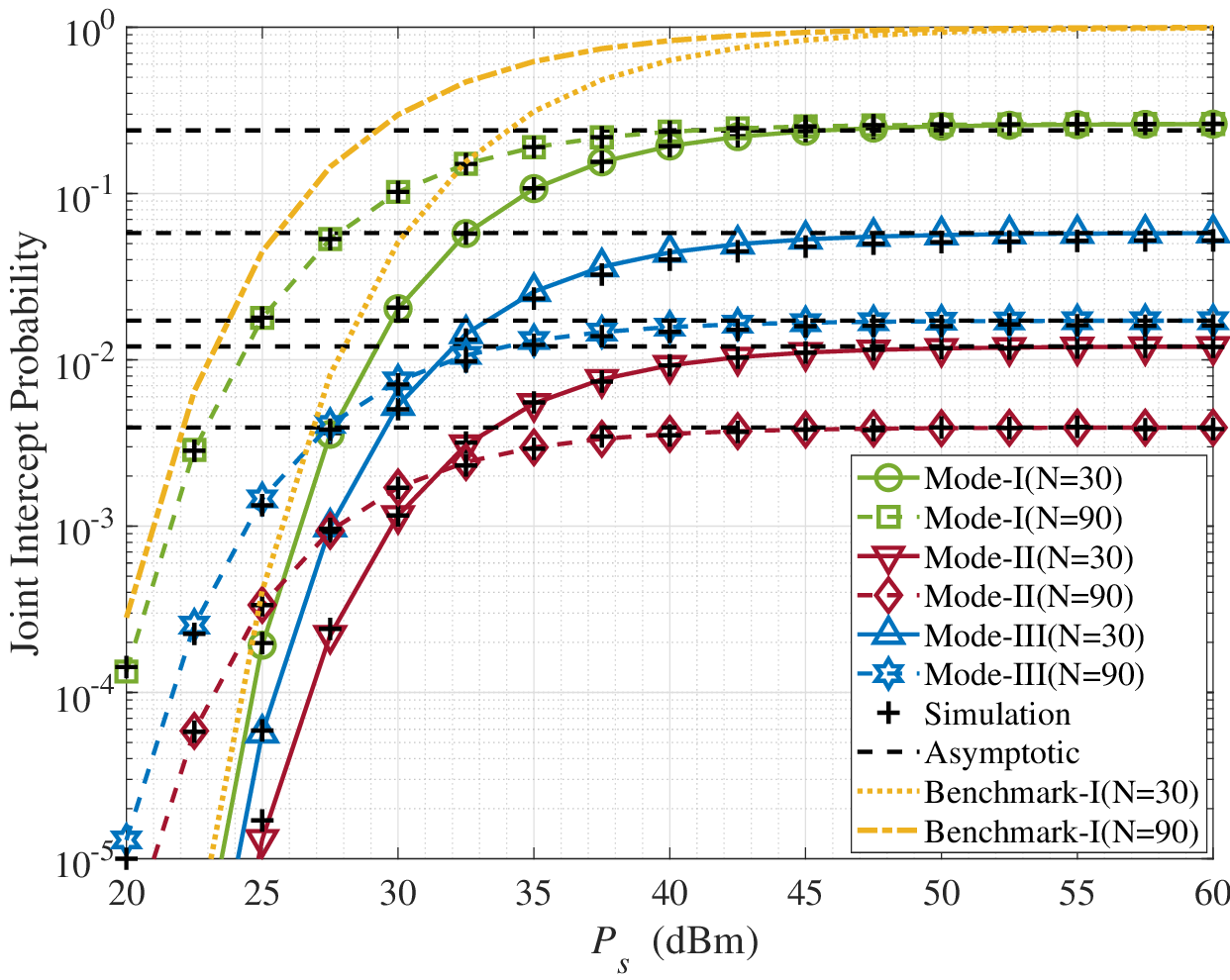}
    \caption{The JIP versus $P_s$ for proposed modes under different $N$.}
    \label{fig3}
\end{figure}

Fig. \ref{fig3} shows that the JIP versus transmission power $P_s$ for proposed modes.
From this figure, we see that the JIPs of proposed three modes are lower than benchmark-\Rmnum{1} under the same conditions, demonstrating the importance of the friendly J in improving system security.
It can be seen that the JIP of mode-\Rmnum{3} is higher than that of mode-\Rmnum{2} but lower than that of mode-\Rmnum{1}, since the number of zeRIS elements that adopt an optimal phase shift design with respect to J is 0, $N$, and $N_2$ in the three modes, respectively, which validates the conclusion in \textit{Remark \ref{remark6}}.
Besides, it is observed from the figure that in mode-\Rmnum{1}, increasing the number of zeRIS elements $N$ degrades JIP performance.
This result can be explained by the fact that increasing $N$ reduces the power consumption of the zeRIS controller shared by each zeRIS element $\frac{P_c}{N}$ in (\ref{C}) and enhances the eavesdropping signal at E, which increases the probability of energy sufficiency event and data interception event.
In addition, in mode-\Rmnum{2} and mode-\Rmnum{3}, increasing $N$ leads to an increase in JIP under low $P_s$, while increasing $N$ leads to a decrease in JIP under high $P_s$.
As discussed, this is due to the fact that under low $P_s$, the impact of increasing $N$ on reducing $\frac{P_c}{N}$ plays a dominant role in the JIP, which indicates that the probability of energy sufficiency event grows to a greater extent than the probability of data interception event, thus increasing the JIP.
Nevertheless, under high $P_s$, increasing $N$ makes the artificial noise at E stronger, and further decreases the probability of data interception event.
This enhancement outweighs the negative effect of an increase in the probability of energy sufficiency event, thus improving security of the system.
Moreover, it should be mentioned that the asymptotic performance of mode-\Rmnum{1} is independent of $N$, while the performance floors in mode-\Rmnum{2} and mode-\Rmnum{3} can be enhanced by increasing $N$, corroborating the conclusion in \textit{Remark \ref{remark7}}.

\begin{figure}[t]
    \centering
    \includegraphics[width=0.6\linewidth]{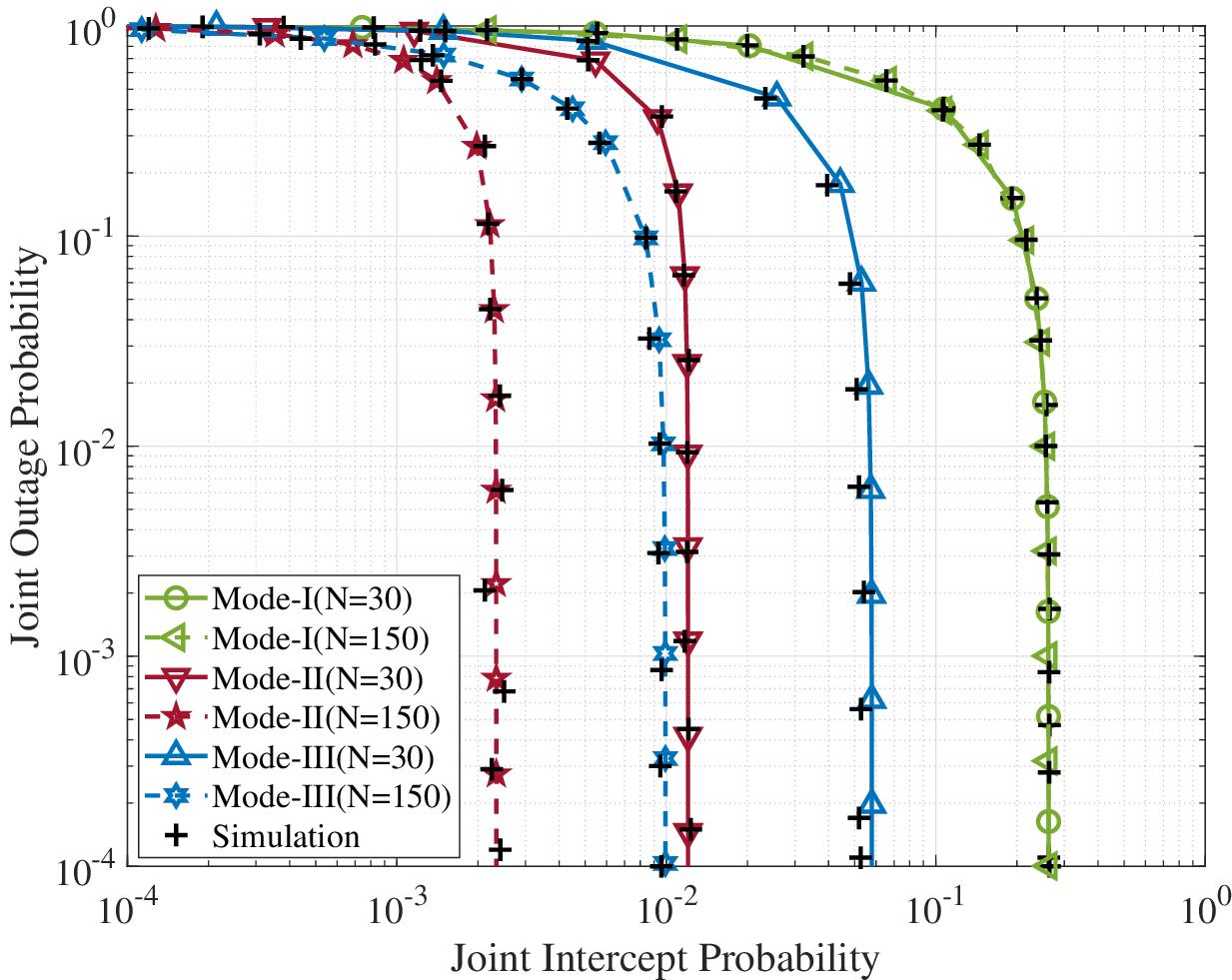}
    \caption{The JOP versus JIP for proposed modes under different $N$.}
    \label{fig6}
\end{figure}

\begin{figure}[t]
    \centering
    \includegraphics[width=0.6\linewidth]{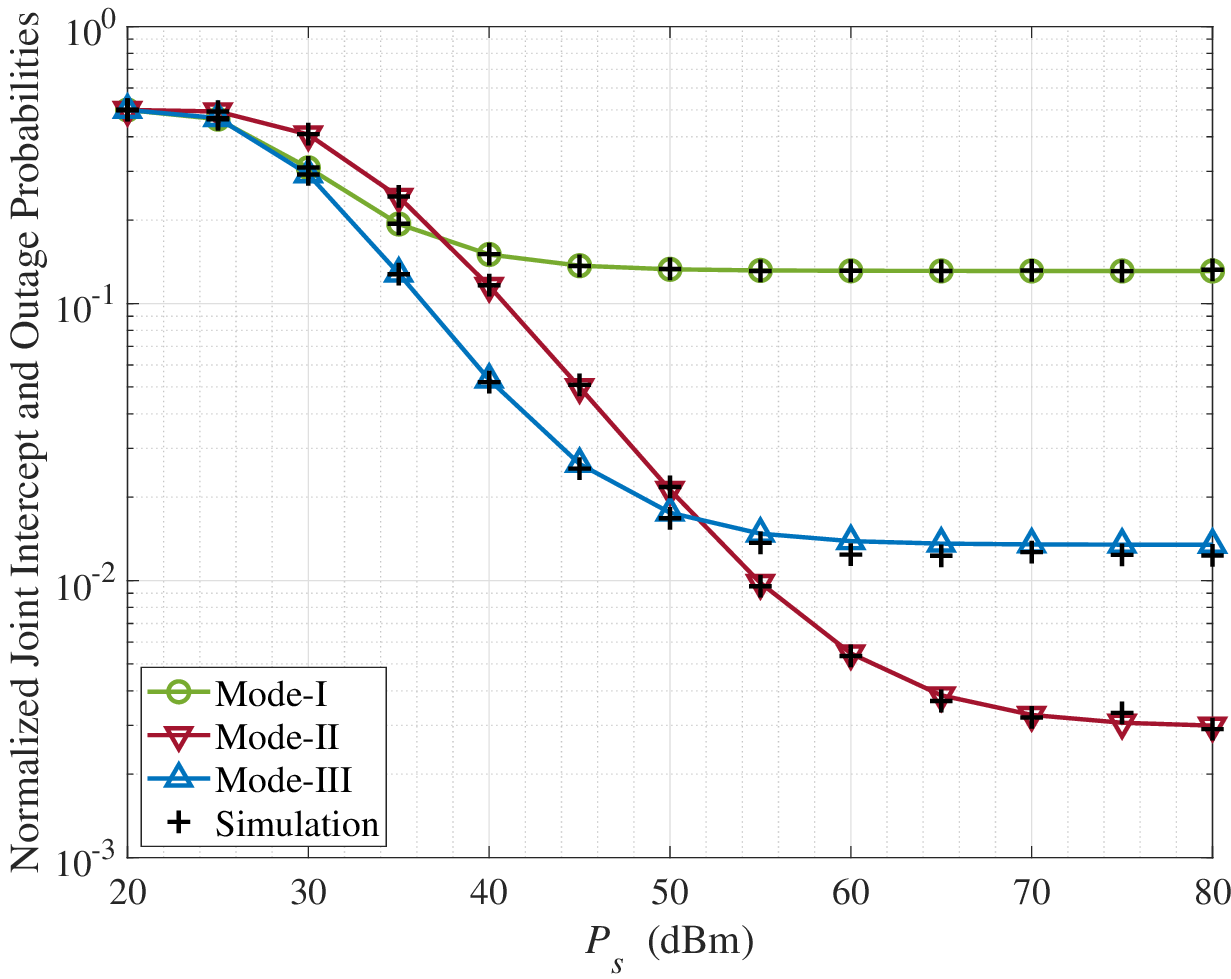}
    \caption{Normalized JIOP versus $P_s$ for proposed modes, where $N=60$.}
    \label{fig7}
\end{figure}

\begin{figure}[t]
    \centering
    \includegraphics[width=0.6\linewidth]{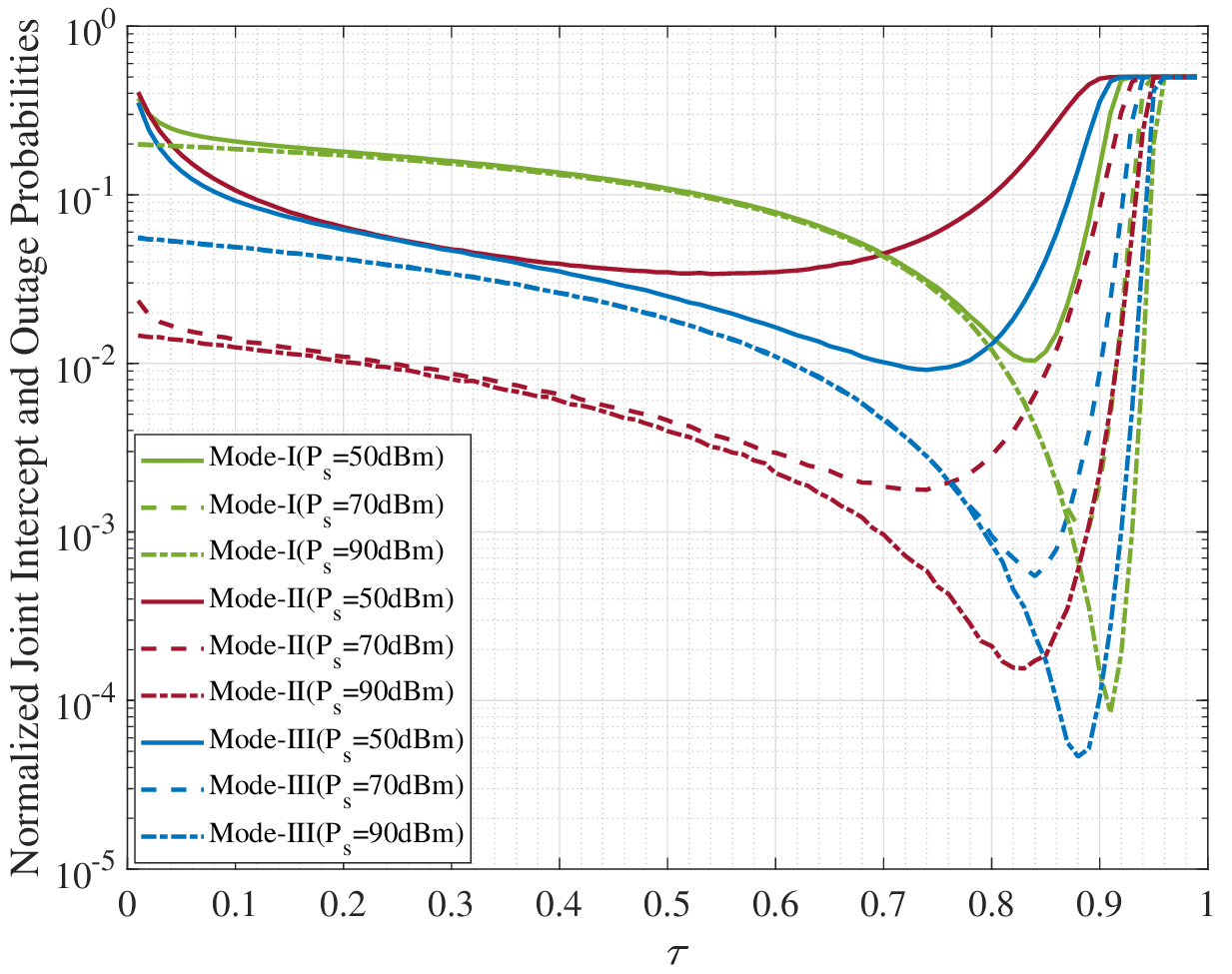}
    \caption{Normalized JIOP versus $\tau$ for proposed modes under different $P_s$.}
    \label{fig8}
\end{figure}

Fig. \ref{fig6} plots that the JOP versus JIP for proposed modes under different number of zeRIS elements $N$.
It is worth noting that the JOPs of all modes decrease as the JIPs increase, and vice versa.
This inverse relationship emphasizes the trade-off between the security and reliability of the system, this metric hereinafter referred to as the SRT, which is related to the fact that the signals received by both legitimate user and eavesdropper are enhanced by the increase in $P_s$.
Under the same condition, the SRT curve of mode-\Rmnum{2} is lower than the other two modes, followed by mode-\Rmnum{3} and mode-\Rmnum{1}.
This result means that mode-\Rmnum{2} achieves the best SRT performance, but this is achieved at the cost of higher system overhead, where all zeRIS elements $N$ in mode-\Rmnum{2} are required to acquire the CSI of E, while only $N_1$ zeRIS elements in mode-\Rmnum{3} are required to acquire the CSI of E, suggesting a trade-off between the SRT performance and the system overhead.
Furthermore, more pre-knowledge of the continuously updated CSI significantly increases system overhead.
It can be observed from the figure that the SRT curve of mode-\Rmnum{1} remains nearly constant with increasing $N$, while that of the other two modes become lower.
Additionally, mode-\Rmnum{3} with $N=150$ achieves a better SRT than mode-\Rmnum{2} with $N=30$, implying that the performance gap between modes can be bridged by increasing $N$.
The above results confirm the conclusions in Fig. \ref{fig2} and Fig. \ref{fig3} once again.

Fig. \ref{fig7} reveals that normalized JIOP, another performance metric to illustrate SRT performance, versus transmission power $P_s$ for proposed modes.
As the normalized JIOP decreases, the SRT performance increases.
Based on the results in Fig. \ref{fig2} and Fig. \ref{fig3}, normalized JIOPs of all three modes are equal to $\frac{1}{2}$ at low $P_s$ and half of their respective JIPs at high $P_s$, where the JIPs and JOPs can be ignored respectively at low and high $P_s$ due to their exponentially small values.
An unexpected phenomenon is that among the three modes, mode-\Rmnum{3} has the lowest normalized JIOP when $P_s$ is below 51.7 dBm, while normalized JIOP of mode-\Rmnum{2} is the lowest when $P_s$ is above 51.7 dBm, demonstrating the superiority of mode-\Rmnum{3} and mode-\Rmnum{2} at low and high $P_s$, respectively.
Moreover, normalized JIOP of mode-\Rmnum{2} is even higher than that of mode-\Rmnum{1} with $P_s$ below 37.6 dBm.
The reason for these results is that when $P_s$ is less than 51.7 dBm, normalized JIOP is dominated by JOP, while normalized JIOP is dominated by JIP conversely.

\begin{figure}[t]
    \centering
    \includegraphics[width=0.6\linewidth]{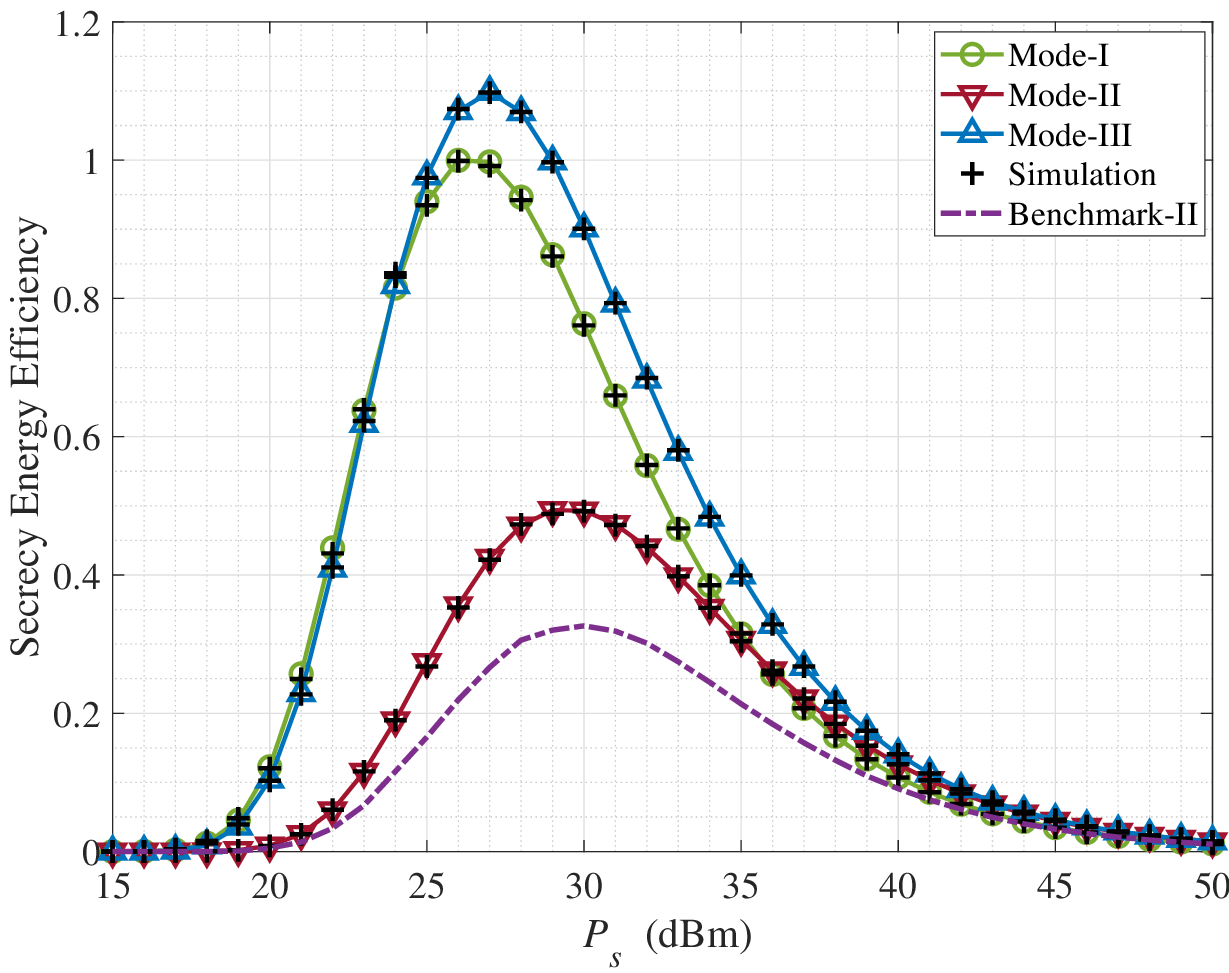}
    \caption{The SEE versus $P_s$ for proposed modes, where $N=100$.}
    \label{fig9}
\end{figure}

Fig. \ref{fig8} demonstrates that normalized JIOP versus time allocation factor $\tau$ for proposed modes under different $P_s$.
In particular, normalized JIOPs of proposed three modes first fall to their respective minimums and then start to grow in the cases of $P_s=50,70,90$ dBm, indicating that each mode possesses an optimal time allocation factor $\tau^{opt}$ for minimizing normalized JIOPs between ET stage and IT stage, respectively.
One can observe from the figure that under the same $P_s$, normalized JIOP corresponding to $\tau^{opt}$ of mode-\Rmnum{3} is the lowest, followed by mode-\Rmnum{1} and mode-\Rmnum{2} orderly, implying that mode-\Rmnum{3} can achieve the best SRT performance, and thus mode-\Rmnum{3} is the optimal mode among the three modes.
Besides, when $P_s$ varies from low to high, the above conclusion is still valid, and the optimal time allocation factors of all three modes appear later with increasing $\tau$.
The underlying reason for this result is that under high $P_s$, the JOP is low enough, and normalized JIOP is dominated by JIP, so that a smaller proportion of time needs to be allocated to IT stage $(1-\tau)T$ to reduce the eavesdropping signal transmission time, thereby decreasing the JIP and guaranteeing the security of the system.
As it can be observed, when $\tau$ is too high, all curves converge to $\frac{1}{2}$.
This result is in light of the fact that in the high $\tau$, the most of time is allocated to ET stage, so the probability of energy outage event and energy sufficiency event converge to 0 and 1, respectively.
In addition, almost no time is allocated to the IT stage, resulting in the probability of data outage event and data interception event converging to 1 and 0, respectively.

Fig. \ref{fig9} depicts that the SEE versus transmission power $P_s$ for proposed modes.
As shown in the figure, the SEE of proposed three modes is higher than benchmark-\Rmnum{2}, showing the superiority of proposed three modes in terms of security, reliability and energy efficiency.
Mode-\Rmnum{3} has the highest SEE among three modes, followed by mode-\Rmnum{1} and mode-\Rmnum{2} orderly, showing the superiority of mode-\Rmnum{3}.
Besides, the SEE of all three modes first increases and then decreases with the increase of $P_s$.
This is because that, as $P_s$ increases, JOP decreases while JIP increases.
At low $P_s$, the enhanced effect of decreasing JOP on SEE is stronger than the deteriorated effect of increasing JIP and $P_s$ on SEE, while the opposite is true at high $P_s$.

\section{Conclusion}
In this paper, we introduced the concept of wireless-powered zeRIS, and presented a wireless-powered zeRIS aided communication system in terms of security, reliability and energy efficiency.
We proposed three new wireless-powered zeRIS modes and three new metrics, i.e., JOP, JIP, and SEE, and derived their closed-form expressions, respectively.

Particularly, our work unveils that under high transmission power, all the diversity gains of three modes are 1, and the JOPs of mode-\Rmnum{1}, mode-\Rmnum{2} and mode-\Rmnum{3} are improved by increasing the number of zeRIS elements, which are related to $N^2$, $N$, and $N_1^2$, respectively.
In addition, from the perspective of reliability, mode-\Rmnum{1} achieves the best JOP, while mode-\Rmnum{2} achieves the best JIP from the perspective of security.
We exploit two SRT metrics, i.e., JOP versus JIP and normalized JIOP, to analyze the SRT performance of the proposed three modes.
For the JOP versus JIP, mode-\Rmnum{2} outperforms the other two modes.
Moreover, for normalized JIOP, mode-\Rmnum{3} and mode-\Rmnum{2} achieve the best SRT performance at low and high transmission power, respectively.
Surprisingly, the SEE of mode-\Rmnum{3} is the highest with increasing transmission power.


\renewcommand\thesubsectiondis{\Roman{subsection}.}

\setcounter{equation}{0}
\renewcommand{\theequation}{\thesection.\arabic{equation}}

\begin{appendices}
\section{Proof of Lemma 1}\label{AppendixA}

Let us denote $N_{1}+1=\lambda$, $\Delta_{1}$ can be rewritten as
\begin{align}
\label{Delta1}
&\Delta_{1}={\left\lvert \sum_{n_1 = 1}^{N_1}  \xi_{n_1} + \sum_{n_2 = \lambda}^{N} \xi_{n_2}    e^{j\omega_{n_2}} \right\rvert }^{2} \nonumber
\\
&=\underbrace{\sum_{n_1 = 1}^{N_1} \xi_{n_1}^2}_{\Lambda_1} + \underbrace{\sum_{n_1 = 1}^{N_1}\sum_{\substack{m_1 = 1 \\ m_1\neq n_1}}^{N_1} \xi_{n_1} \xi_{m_1}}_{\Lambda_2}  +\underbrace{\sum_{n_1 = 1}^{N_1}\sum_{n_2 = \lambda}^{N} \xi_{n_1} \xi_{n_2}   \cos\omega_{n_2}}_{\Lambda_3} \nonumber
\\
& + \underbrace{\sum_{n_2 = \lambda}^{N}   \xi_{n_2}^2}_{\Lambda_4} +\underbrace{\sum_{n_2 = \lambda}^{N} \sum_{\substack{m_2 = \lambda\\m_2\neq n_2}}^{N}  \xi_{n_2}   \xi_{m_2} \cos (\omega_{n_2}-\omega_{m_2}) }_{\Lambda_5}.
\end{align}

Hence, the first moment of $\Delta_{1}$ can be calculated as
\begin{align}
    \label{EDelta1}
    &u_{\Delta_{1}}=\mathbb{E} \{\Delta_{1}\} \nonumber
    \\
    &=\sum_{n_1 = 1}^{N_1} \mathbb{E} \{\xi_{n_1}^2\} + \sum_{n_1 = 1}^{N_1}\sum_{\substack{m_1 = 1 \\ m_1\neq n_1}}^{N_1} \mathbb{E} \{\xi_{n_1}\} \mathbb{E} \{\xi_{m_1}\} \nonumber
    \\
    &+2\sum_{n_1 = 1}^{N_1}\sum_{n_2 = \lambda}^{N} \mathbb{E} \{\xi_{n_1}\} \mathbb{E} \{\xi_{n_2}\}  \mathbb{E} \{ \cos\omega_{n_2}\} +\sum_{n_2 = \lambda}^{N}  \mathbb{E} \{ \xi_{n_2}^2\} \nonumber
    \\
    &+\sum_{n_2 = \lambda}^{N} \sum_{\substack{m_2 = \lambda\\m_2\neq n_2}}^{N} \mathbb{E} \{ \xi_{n_2} \} \mathbb{E} \{ \xi_{m_2}\}\mathbb{E} \{ \cos (\omega_{n_2}-\omega_{m_2}) \}.
\end{align}

\vspace{-0.5em}
Moreover, $\Delta_{1}^2$ can be expressed as
\begin{align}
    \label{Delta12}
    \Delta_{1}^2=&(\Lambda_1+\Lambda_2+\Lambda_3+\Lambda_4+\Lambda_5)^2 \nonumber
    \\
    =&\Lambda_1^2 + \Lambda_2^2 + \Lambda_3^2 + \Lambda_4^2 + \Lambda_5^2 + 2\Big(\Lambda_1\Lambda_2 + \Lambda_1\Lambda_3 + \Lambda_1\Lambda_4 \nonumber
    \\
    &+ \Lambda_1\Lambda_5 + \Lambda_2\Lambda_3 + \Lambda_2\Lambda_4 + \Lambda_2\Lambda_5 + \Lambda_3\Lambda_4 + \Lambda_3\Lambda_5 \nonumber
    \\
    &+ \Lambda_4\Lambda_5\Big).
\end{align}

\vspace{-0.5em}
The expectation of each term in (\ref{Delta12}) is given by
\begin{align}
    \label{1*1}
    &\mathbb{E}\{ \Lambda_1^2 \}=\sum_{n_1 = 1}^{N_1} \mathbb{E} \{ \xi_{n_1}^4 \} + \sum_{n_1 = 1}^{N_1}\sum_{\substack{m_1 = 1 \\ m_1\neq n_1}}^{N_1} \hspace{-0.4em} \mathbb{E} \{\xi_{n_1}^2\} \mathbb{E} \{\xi_{m_1}^2\},
    \\
    \label{2*2}
    &\mathbb{E}\{ \Lambda_2^2 \}= 2\sum_{n_1 = 1}^{N_1}\sum_{\substack{m_1 = 1 \\ m_1\neq n_1}}^{N_1} \mathbb{E} \{\xi_{n_1}^2\} \mathbb{E} \{\xi_{m_1}^2\} \nonumber
    \\
    &+ 2\sum_{n_1 = 1}^{N_1}\sum_{\substack{m_1 = 1 \\ m_1\neq n_1}}^{N_1} \sum_{\substack{p_1 = 1 \\ p_1\neq n_1, m_1}}^{N_1} \mathbb{E} \{\xi_{n_1}^2\} \mathbb{E} \{\xi_{m_1}\} \mathbb{E} \{\xi_{p_1}\} \nonumber
    \\
    &+ 2\sum_{n_1 = 1}^{N_1}\sum_{\substack{m_1 = 1 \\ m_1\neq n_1}}^{N_1} \sum_{\substack{p_1 = 1 \\ p_1\neq n_1, m_1}}^{N_1} \mathbb{E} \{\xi_{n_1}\} \mathbb{E} \{\xi_{m_1}^2\} \mathbb{E} \{\xi_{p_1}\} \nonumber
    \\
    &+ \sum_{n_1 = 1}^{N_1}\sum_{\substack{m_1 = 1 \\ m_1\neq n_1}}^{N_1} \sum_{\substack{p_1 = 1 \\ p_1\neq n_1, m_1}}^{N_1} \sum_{\substack{q_1 = 1 \\ q_1\neq n_1, m_1,p_1}}^{N_1} \mathbb{E} \{\xi_{n_1}\} \mathbb{E} \{\xi_{m_1}\} \mathbb{E} \{\xi_{p_1}\} \nonumber 
    \\
    &\hspace{1em}\times \mathbb{E} \{\xi_{q_1}\},
    \\
    \label{3*3}
    &\mathbb{E}\{ \Lambda_3^2 \}=4\Bigg( \sum_{n_1 = 1}^{N_1} \sum_{n_2 = \lambda}^{N}  \mathbb{E} \{\xi_{n_1}^2\}\mathbb{E} \{\xi_{n_2}^2\} \mathbb{E} \{\cos^{2} \omega_{n_2}\} \nonumber
    \\
    &+ \sum_{n_1 = 1}^{N_1} \sum_{\substack{m_1 = 1 \\ m_1\neq n_1}}^{N_1} \sum_{n_2 = \lambda}^{N}  \mathbb{E} \{\xi_{n_1}\} \mathbb{E} \{\xi_{m_1}\}  \mathbb{E} \{\xi_{n_2}^2\} \mathbb{E} \{\cos^{2} \omega_{n_2}\} \nonumber
    \\
    &+ \sum_{n_1 = 1}^{N_1} \sum_{n_2 = \lambda}^{N} \sum_{\substack{m_2 = \lambda \\ m_2\neq n_2}}^{N}  \mathbb{E} \{\xi_{n_1}^2\} \mathbb{E} \{\xi_{n_2}\}  \mathbb{E} \{\xi_{m_2}\} \mathbb{E} \{\cos \omega_{n_2}\}  \nonumber
    \\
    &\hspace{1em} \times \mathbb{E} \{\cos \omega_{m_2}\} \nonumber
    \\
    &+ \sum_{n_1 = 1}^{N_1} \sum_{\substack{m_1 = 1 \\ m_1\neq n_1}}^{N_1} \sum_{n_2 = \lambda}^{N} \sum_{\substack{m_2 = \lambda \\ m_2\neq n_2}}^{N}  \mathbb{E} \{\xi_{n_1}\} \mathbb{E} \{\xi_{m_1}\}  \mathbb{E} \{\xi_{n_2}\} \mathbb{E} \{\xi_{m_2}\}  \nonumber
    \\
    &\hspace{1em} \times  \mathbb{E} \{\cos \omega_{n_2}\} \mathbb{E} \{\cos \omega_{m_2}\} \Bigg),
    \\
    \label{4*4}
    &\mathbb{E}\{ \Lambda_4^2 \}= \sum_{n_2 = \lambda}^{N} \mathbb{E} \{\xi_{n_2}^4\} + \sum_{n_2 = \lambda}^{N} \sum_{\substack{m_2 = \lambda \\ m_2\neq n_2}}^{N} \hspace{-0.5em} \mathbb{E} \{\xi_{n_2}^2\}\mathbb{E} \{\xi_{m_2}^2\},
    \\
    \label{5*5}
    &\mathbb{E}\{ \Lambda_5^2 \}=2\sum_{n_2 = \lambda}^{N} \sum_{\substack{m_2 = \lambda \\ m_2\neq n_2}}^{N} \mathbb{E} \{\xi_{n_2}^2\}\mathbb{E} \{\xi_{m_2}^2\} \mathbb{E} \{\cos^2 (\omega_{n_2}-\omega_{m_2})\} \nonumber
    \\
    &+2\sum_{n_2 = \lambda}^{N} \sum_{\substack{m_2 = \lambda \\ m_2\neq n_2}}^{N} \sum_{\substack{p_2 = \lambda \\ p_2\neq n_2,m_2}}^{N} \mathbb{E} \{\xi_{n_2}^2\}\mathbb{E} \{\xi_{m_2}\}\mathbb{E} \{\xi_{p_2}\} \nonumber
    \\
    &\hspace{1em} \times \mathbb{E} \{\cos (\omega_{n_2}-\omega_{m_2}) \cos (\omega_{n_2}-\omega_{p_2})\} \nonumber
    \\
    &+2\sum_{n_2 = \lambda}^{N} \sum_{\substack{m_2 = \lambda \\ m_2\neq n_2}}^{N} \sum_{\substack{p_2 = \lambda \\ p_2\neq n_2,m_2}}^{N} \mathbb{E} \{\xi_{n_2}\}\mathbb{E} \{\xi_{m_2}^2\}\mathbb{E} \{\xi_{p_2}\} \nonumber
    \\
    &\hspace{1em} \times \mathbb{E} \{\cos (\omega_{n_2}-\omega_{m_2}) \cos (\omega_{m_2}-\omega_{p_2})\} \nonumber
    \\
    &+\sum_{n_2 = \lambda}^{N} \hspace{-0.1em} \sum_{\substack{m_2 = \lambda \\ m_2\neq n_2}}^{N} \hspace{-0.1em} \sum_{\substack{p_2 = \lambda \\ p_2\neq n_2,m_2}}^{N} \hspace{-0.1em}\sum_{\substack{q_2 = \lambda \\ q_2\neq n_2,m_2,p_2}}^{N} \hspace{-1.5em} \mathbb{E} \{\xi_{n_2}\}\mathbb{E} \{\xi_{m_2}\}\mathbb{E} \{\xi_{p_2}\} \mathbb{E} \{\xi_{q_2}\}  \nonumber
    \\
    &\hspace{1em} \times \mathbb{E} \{\cos (\omega_{n_2}-\omega_{m_2})\} \mathbb{E} \{\cos (\omega_{p_2}-\omega_{q_2})\},
    \\
    \label{1*2}
    &\mathbb{E}\{ \Lambda_1\Lambda_2 \}= \sum_{n_1 = 1}^{N_1}\sum_{\substack{m_1 = 1 \\ m_1\neq n_1}}^{N_1} \mathbb{E} \{\xi_{n_1}^3\} \mathbb{E} \{\xi_{m_1}\} \nonumber
    \\
    & + \sum_{n_1 = 1}^{N_1}\sum_{\substack{m_1 = 1 \\ m_1\neq n_1}}^{N_1} \sum_{\substack{p_1 = 1 \\ p_1\neq n_1,m_1}}^{N_1} \mathbb{E} \{\xi_{n_1}^2\} \mathbb{E} \{\xi_{m_1}\} \mathbb{E} \{\xi_{p_1}\},
    \\
    \label{1*3}
    &\mathbb{E}\{ \Lambda_1\Lambda_3 \}=2\Bigg( \sum_{n_1 = 1}^{N_1} \sum_{n_2 = \lambda}^{N} \mathbb{E} \{\xi_{n_1}^3\} \mathbb{E} \{\xi_{n_2}\} \mathbb{E} \{\cos \omega_{n_2}\} \nonumber
    \\
    &+ \sum_{n_1 = 1}^{N_1}\sum_{\substack{m_1 = 1 \\ m_1\neq n_1}}^{N_1} \sum_{n_2 = \lambda}^{N} \mathbb{E} \{\xi_{n_1}^2\} \mathbb{E} \{\xi_{m_1}\} \mathbb{E} \{\xi_{n_2}\} \mathbb{E} \{\cos \omega_{n_2}\} \Bigg),
    \\
    \label{1*4}
    &\mathbb{E}\{ \Lambda_1\Lambda_4 \}=\sum_{n_1 = 1}^{N_1} \sum_{n_2 = \lambda}^{N}\mathbb{E} \{\xi_{n_1}^2\}\mathbb{E} \{\xi_{n_2}^2\},
    \\
    \label{1*5}
    &\mathbb{E}\{ \Lambda_1\Lambda_5 \}=\sum_{n_1 = 1}^{N_1} \sum_{n_2 = \lambda}^{N} \sum_{\substack{m_2 = \lambda \\ m_2\neq n_2}}^{N} \hspace{-0.35em} \mathbb{E} \{\xi_{n_1}^2\}  \mathbb{E} \{\xi_{n_2}\} \mathbb{E} \{\xi_{m_2}\},
    \\
    \label{2*3}
    &\mathbb{E}\{ \Lambda_2\Lambda_3 \}=2\Bigg( \sum_{n_1 = 1}^{N_1}\sum_{\substack{m_1 = 1 \\ m_1\neq n_1}}^{N_1} \sum_{\substack{p_1 = 1 \\ p_1\neq n_1,m_1}}^{N_1} \sum_{n_2 = \lambda}^{N}  \mathbb{E} \{\xi_{n_1}\} \mathbb{E} \{\xi_{m_1}\}  \nonumber
    \\
    &\hspace{1em} \times \mathbb{E} \{\xi_{p_1}\} \mathbb{E} \{\xi_{n_2}\}  \mathbb{E} \{\cos \omega_{n_2}\} \nonumber
    \\
    &+2 \sum_{n_1 = 1}^{N_1}\sum_{\substack{m_1 = 1 \\ m_1\neq n_1}}^{N_1} \sum_{n_2 = \lambda}^{N} \mathbb{E} \{\xi_{n_1}^2\} \mathbb{E} \{\xi_{m_1}\} \mathbb{E} \{\xi_{n_2}\}  \mathbb{E} \{\cos \omega_{n_2}\} \Bigg),
    \\
    \label{2*4}
    &\mathbb{E}\{ \Lambda_2\Lambda_4 \}= \sum_{n_1 = 1}^{N_1}\sum_{\substack{m_1 = 1 \\ m_1\neq n_1}}^{N_1} \sum_{n_2 = \lambda}^{N} \mathbb{E} \{\xi_{n_1}\} \mathbb{E} \{\xi_{m_1}\} \mathbb{E} \{\xi_{n_2}^2\},
    \\
    \label{2*5}
    &\mathbb{E}\{ \Lambda_2\Lambda_5 \}= \sum_{n_1 = 1}^{N_1}\sum_{\substack{m_1 = 1 \\ m_1\neq n_1}}^{N_1} \sum_{n_2 = \lambda}^{N} \sum_{\substack{m_2 = \lambda \\ m_2\neq n_2}}^{N} \mathbb{E} \{\xi_{n_1}\} \mathbb{E} \{\xi_{m_1}\} \mathbb{E} \{\xi_{n_2}\} \nonumber
    \\
    & \hspace{1em} \times \mathbb{E} \{\xi_{m_2}\} \mathbb{E} \{\cos (\omega_{n_2}-\omega_{m_2})\},
    \\
    \label{3*4}
    &\mathbb{E}\{ \Lambda_3\Lambda_4 \}=2\Bigg(  \sum_{n_1 = 1}^{N_1} \sum_{n_2 = \lambda}^{N} \mathbb{E} \{\xi_{n_1}\} \mathbb{E} \{\xi_{n_2}^3\} \mathbb{E} \{\cos\omega_{n_2}\} \nonumber
    \\
    &+ \sum_{n_1 = 1}^{N_1}  \sum_{n_2 = \lambda}^{N} \sum_{\substack{m_2 = \lambda \\ m_2\neq n_2}}^{N} \mathbb{E} \{\xi_{n_1}\} \mathbb{E} \{\xi_{n_2}\} \mathbb{E} \{\xi_{m_2}^2\} \mathbb{E} \{\cos\omega_{n_2}\} \Bigg),
    \\
    \label{3*5}
    &\mathbb{E}\{ \Lambda_3\Lambda_5 \}=2\Bigg( \sum_{n_1 = 1}^{N_1} \sum_{n_2 = \lambda}^{N} \sum_{\substack{m_2 = \lambda \\ m_2\neq n_2}}^{N} \sum_{\substack{p_2 = \lambda \\ p_2\neq n_2,m_2}}^{N} \mathbb{E} \{\xi_{n_1}\} \mathbb{E} \{\xi_{n_2}\} \nonumber
    \\
    &\hspace{1em} \times \mathbb{E} \{\xi_{m_2}\} \mathbb{E} \{\xi_{p_2}\} \mathbb{E} \{\cos\omega_{n_2}\} \mathbb{E} \{\cos (\omega_{m_2}-\omega_{p_2})\} \nonumber 
    \\
    &+ \sum_{n_1 = 1}^{N_1}  \sum_{n_2 = \lambda}^{N} \sum_{\substack{m_2 = \lambda \\ m_2\neq n_2}}^{N} \mathbb{E} \{\xi_{n_1}\} \mathbb{E} \{\xi_{n_2}^2\} \mathbb{E} \{\xi_{m_2}\} \mathbb{E} \{\cos\omega_{n_2}\} \nonumber
    \\
    &\hspace{1em} \times \mathbb{E} \{\cos (\omega_{n_2}-\omega_{m_2})\}  \nonumber
    \\
    &+ \sum_{n_1 = 1}^{N_1}  \sum_{n_2 = \lambda}^{N} \sum_{\substack{m_2 = \lambda \\ m_2\neq n_2}}^{N} \mathbb{E} \{\xi_{n_1}\} \mathbb{E} \{\xi_{n_2}\} \mathbb{E} \{\xi_{m_2}^2\} \mathbb{E} \{\cos\omega_{m_2}\} \nonumber
    \\
    &\hspace{1em} \times \mathbb{E} \{\cos (\omega_{n_2}-\omega_{m_2})\}   \Bigg),
    \\
    \label{4*5}
    &\mathbb{E}\{ \Lambda_4\Lambda_5 \}=\sum_{n_2 = \lambda}^{N} \sum_{\substack{m_2 = \lambda \\ m_2\neq n_2}}^{N} \mathbb{E} \{\xi_{n_2}^3\} \mathbb{E} \{\xi_{m_2}\}\mathbb{E} \{\cos (\omega_{n_2}-\omega_{m_2})\} \nonumber
    \\
    &+\sum_{n_2 = \lambda}^{N} \sum_{\substack{m_2 = \lambda \\ m_2\neq n_2}}^{N} \mathbb{E} \{\xi_{n_2}\} \mathbb{E} \{\xi_{m_2}^3\}\mathbb{E} \{\cos (\omega_{n_2}-\omega_{m_2})\}.
\end{align}

Due to the fact that $\left\lvert h_{ur,n_1} \right\rvert$, $\left\lvert h_{ra,n_1} \right\rvert$, $\left\lvert h_{ur,n_2} \right\rvert$ and $\left\lvert h_{ra,n_2} \right\rvert$ follow the identical Rayleigh distributions, the moments of $\xi_{n_1}$ and $\xi_{n_2}$ are given by
\begin{align}
    \label{R1}
    \mathbb{E} \{\xi_{n_1}\} &= \mathbb{E} \{\xi_{n_2}\} = \frac{\pi}{4}, \ \ \forall n_{1}, n_{2} ,
\\
    \label{R2}
    \mathbb{E} \{\xi_{n_1}^2\} &= \mathbb{E} \{\xi_{n_2}^2\} = 1, \ \ \forall n_{1}, n_{2} ,
\\
    \label{R3}
    \mathbb{E} \{\xi_{n_1}^3\} &= \mathbb{E} \{\xi_{n_2}^3\} = \frac{9}{16}\pi, \ \ \forall n_{1}, n_{2} ,
\\
    \label{R4}
    \mathbb{E} \{\xi_{n_1}^4\} &= \mathbb{E} \{\xi_{n_2}^4\} = 4, \ \ \forall n_{1}, n_{2} .
\end{align}

According to the fact that $\omega_{n_2},\omega_{m_2},\omega_{p_2} \sim U(-\pi, \pi)$ and \cite[Lemma 6]{9599656}, we have
\begin{align}
    \label{cos1}
    &\mathbb{E} \{\cos\omega_{n_2}\} = 0,
    \\
    \label{cos2}
    &\mathbb{E} \{\cos(\omega_{n_2}-\omega_{m_2})\} = 0,\ \ \forall m_2 \neq n_2 ,
    \\
    \label{cos3}
    &\mathbb{E} \{\cos^{2}(\omega_{n_2}-\omega_{m_2})\} = \frac{1}{2},\ \ \forall m_2 \neq n_2 ,
    \\
    &\mathbb{E} \{\cos(\omega_{n_2}-\omega_{m_2})\cos(\omega_{n_2}-\omega_{p_2})\} = 0,\nonumber
    \\
    \label{cos4}
    &\hspace{10em}\forall m_2 \neq n_2, p_2 \neq n_2, m_2 .
\end{align}

Based on (\ref{R1})-(\ref{cos4}), the expectations of $\Delta_{1}$ and $\Delta_{1}^2$ are respectively derived as
\begin{align}
    \label{EDelta1_final}
    &\mathbb{E}\{ \Delta_{1} \} = N + N_{1}(N_{1}-1)\frac{\pi^2}{16},
    \\
    \label{EDelta12_final}
    &\mathbb{E}\{ \Delta_{1}^2 \} =\frac{1}{256}\bigg( \Big (32 \pi^2 N_{1}+512 \Big ) N_2^2 + \Big (64 \pi^2 N_1^2 + 512 \nonumber
    \\
    &+ \big (1024-96\pi^2 \big ) N_1 \Big) N_2 + \pi^4 N_1^4 + \Big(96\pi^2-6\pi^4 \Big) N_1^3 \nonumber
    \\
    & + \Big( 11\pi^4 -216\pi^2+768 \Big)N_1^2 - \Big(6\pi^4 -120\pi^2 -256 \Big)N_1   \bigg). 
\end{align}

Therefore, substituting (\ref{EDelta1_final}) and (\ref{EDelta12_final}) into (\ref{k1theta1}), the shape parameter $k_1$ and scale parameter $\theta_1$ are obtained.

Similarly, exchanging $N_1$ and $N_2$, the distribution of $\Delta_{2}$ is also obtained, and the proof of Lemma \ref{lemma1} is completed.

\setcounter{equation}{0}

\section{JOP of Proposed Modes}\label{AppendixB}
\subsection{Proof of Theorem \ref{theorem1}}
Based on (\ref{JOPi2}), we can prove this theorem in two steps.

    First, we derive $A$.
    Due to the fact that the sum of complex Gaussian RVs still follows complex Gaussian distribution, $X={\beta }_{pr}\left\lvert\sum_{n = 1}^{N}{h}_{pr,n} \right\rvert^{2}$ has an exponential distribution with mean $N{\beta }_{pr}$.
    Therefore, utilizing (\ref{Eris}) and (\ref{Q}), we have
    \begin{align}\label{A}
        A=&{\rm Pr} \bigg[ X < \frac{NP_{e}+P_{c}}{P_t}  \bigg] \nonumber
        \\
        =&1-e^{-\frac{NP_{e}+P_{c}}{P_t N\beta_{pr}}}.
    \end{align}

    Then, before deriving $B^{\Rmnum{1}}$, we need to determine the statistic of zeRIS cascaded channels $Z_{\Rmnum{1}}=\sum_{n = 1}^{N} \left\lvert {h}_{ur,n} \right\rvert \left\lvert {h}_{ra,n} \right\rvert$.
    According to \cite[Lemma 1]{10130095}, the probability density function (PDF) of $Z_{\Rmnum{1}}$ can be denoted as
    \begin{align}
        \label{PDF3}
        f_{Z_{\Rmnum{1}}}(z)=\frac{z^{v-1}}{\Gamma(v)\varphi^{v}}e^{-\frac{z}{\varphi} },
    \end{align}
    where $v=\frac{N\pi^2}{16-\pi^2}$, $\varphi=\frac{16-\pi^2}{4\pi}$, and $\Gamma(\cdot)$ is the Gamma function \cite[Eq. (8.31)]{10.1115/1.3138251}.

    Utilizing (\ref{SNRa11}) and (\ref{PDF3}), we have
    \begin{align}
        \label{B1}
        B^{\Rmnum{1}}=&{\rm Pr}[(1-\tau)\log_2(1+{\rm{SNR}}_{a}^{\Rmnum{1}})<R] \nonumber
        \\
        =&{\rm Pr}\Big[ M_1 Z_{\Rmnum{1}}^2 < \frac{\epsilon}{\rho_t \beta_{ura}} \Big] \nonumber
        \\
        =&1-\int_{0}^{\infty} e^{-\frac{\varsigma}{\rho_t z^2}} \frac{z^{v-1}}{\Gamma(v)\varphi^{v}}e^{-\frac{z}{\varphi} } \,dz \nonumber
        \\
        \overset{(a)}{=}&1-\frac{\pi^2}{4L\Gamma(v)\varphi^v} \sum_{l = 1}^{L} \sqrt{1-\varpi_l^2}\sec^2u_l (\tan u_l)^{v-1} \nonumber
        \\
        &\times \exp \bigg( -\frac{\varsigma }{\rho_t \tan^2u_l}-\frac{\tan u_l}{\varphi}  \bigg),
    \end{align}
    where $M_1=\beta_{pu}\left\lvert h_{pu}\right\rvert^2 $, and step $(a)$ uses Gaussian-Chebyshev quadrature by replacing $z$ with $\tan u_l$ \cite{9858871}.

    Substituting (\ref{A}) and (\ref{B1}) into (\ref{JOPi2}), we derive the last result in (\ref{JOP1}).
    Thus, the proof of Theorem \ref{theorem1} is completed.

\subsection{Proof of Theorem \ref{theorem2}}
    The derivation of $A$ is the same as (\ref{A}).
    In order to determine $B^{\Rmnum{2}}$, the statistic of zeRIS cascaded channels $\sum_{n = 1}^{N} h_{ur,n} h_{ra,n} e^{j\phi_{\Rmnum{2},n}} $ needs to be concluded.
    By referring to \cite[Lemma 2]{9964281}, we can deduce the fact that $\sum_{n = 1}^{N} h_{ur,n} h_{ra,n} e^{j\phi_{\Rmnum{2},n}} $ follows the complex Gaussian distribution with zero mean and variance of $N$.
    Hence, the PDF of $Z_{\Rmnum{2}}={\beta}_{ura}{\left\lvert \sum_{n = 1}^{N} h_{ur,n} h_{ra,n} e^{j\phi_{\Rmnum{2},n}}  \right\rvert }^{2}$ can be denoted as
    \begin{align}
        \label{PDF4}
        f_{Z_{\Rmnum{2}}}(z)&=\frac{1}{N\beta_{ura}} e^{-\frac{z}{N\beta_{ura}}}.
    \end{align}

    Using (\ref{SNRa2}) and (\ref{PDF4}), we have
    \begin{align}
        \label{B2}
        B^{\Rmnum{2}}=&{\rm Pr}[(1-\tau)\log_2(1+{\rm{SNR}}_{a}^{\Rmnum{2}})<R] \nonumber
        \\
        =&{\rm Pr}\Big[ M_1 Z_{\Rmnum{2}} < \frac{\epsilon}{\rho_t} \Big] \nonumber
        \\
        =& 1 - \frac{1}{N\beta_{ura}} \int_{0}^{\infty} e^{-\frac{\epsilon}{\rho_t \beta_{pu} z}-\frac{z}{N\beta_{ura}}} \,dz \nonumber
        \\
        \overset{(b)}{=}&1-\sqrt{\frac{4\varsigma}{\rho_t N} } K_1 \Bigg(\sqrt{\frac{4\varsigma}{\rho_t N} } \Bigg),
    \end{align}
    where step $(b)$ uses \cite[Eq. (3.324.1)]{10.1115/1.3138251}.
    
    Referring to (\ref{A}) and (\ref{B2}), the final result in (\ref{JOP2}) is derived.
    Therefore, we conclude the proof of Theorem \ref{theorem2}.

\subsection{Proof of Theorem \ref{theorem3}}
    The derivation of $A$ is the same as (\ref{A}).
    In mode-\Rmnum{3}, $Z_{\Rmnum{3}}$ follows the same distribution of $\Delta_1$ in Lemma \ref{lemma1}.
    Employing (\ref{SNRa3}) and (\ref{PDF1}), we have
    \begin{align}
        \label{B3}
        &B^{\Rmnum{3}}={\rm Pr}[(1-\tau)\log_2(1+{\rm{SNR}}_{a}^{\Rmnum{3}})<R] \nonumber
        \\
        &={\rm Pr}\Big[ M_1 Z_{\Rmnum{3}} < \frac{\epsilon}{\rho_t \beta_{ura}} \Big] \nonumber
        \\
        &=1 - \int_{0}^{\infty} e^{-\frac{\varsigma}{\rho_t  z}} \frac{z^{k_{1}-1}}{\Gamma(k_1)\theta_{1}^{k_1}}e^{-\frac{z}{\theta_{1}} } \,dz \nonumber
        \\
        &\overset{(c)}{=} 1 - \frac{2}{\Gamma(k_1)} \Big(\frac{\varsigma}{\rho_t \theta_1} \Big)^{\frac{k_1}{2} } K_{k_1}\Big(2\sqrt{\frac{\varsigma}{\rho_t \theta_1}}\Big).
    \end{align}
    where step $(c)$ utilizes \cite[Eq. (3.471.9)]{10.1115/1.3138251}. 
    
    Based on the results of (\ref{A}) and (\ref{B3}), the final result in (\ref{JOP3}) is obtained.
    Hence, the proof of Theorem \ref{theorem3} is completed.

\setcounter{equation}{0}
\section{JIP of Proposed Modes}\label{AppendixC}
\subsection{Proof of Theorem \ref{theorem4}}
Based on (\ref{A}), we have
    \begin{align}
        \label{C}
        C=&{\rm Pr} \bigg[ X \geq \frac{NP_{e}+P_{c}}{P_t}  \bigg] \nonumber
        \\
        =&e^{-\frac{NP_{e}+P_{c}}{P_t N\beta_{pr}}}.
    \end{align}

    Define $Y_{\Rmnum{1},1}=\beta_{ure}{\left\lvert \sum_{n = 1}^{N} h_{ur,n} h_{re,n} e^{j\phi_{\Rmnum{1},n}}  \right\rvert }^{2}$, $Y_{\Rmnum{1},2}=\beta_{jre}{\left\lvert \sum_{n = 1}^{N} h_{jr,n} h_{re,n} e^{j\phi_{\Rmnum{1},n}}  \right\rvert }^{2}$, and $M_2=\beta_{pj}\left\lvert h_{pj}\right\rvert^2$.
    The PDFs of $Y_{\Rmnum{1},1}$ and $Y_{\Rmnum{1},2}$ are similar to (\ref{PDF3}), whose parameters are $N\beta_{ure}$ and $N\beta_{jre}$, respectively.
    Thus, we have
    \begin{align}
        \label{D1}
        D^{\Rmnum{1}}=&{\rm Pr}[(1-\tau)\log_2(1+{\rm{SNR}}_{e}^{\Rmnum{1}}) \geq R] \nonumber
        \\
        =&{\rm Pr}\Big[ M_1 Y_{\Rmnum{1},1} \geq \frac{\epsilon \rho_t M_2 Y_{\Rmnum{1},2} + \epsilon }{\rho_t } \Big] \nonumber
        \\
        =&\frac{1}{N^2 \beta_{pj} \beta_{ure} \beta_{jre}} \int_{0}^{\infty} e^{-\frac{\epsilon}{\rho_t \beta_{pu} y_1}-\frac{y_1}{N\beta_{ure}}} 
        \\
        &\times \underbrace{\int_{0}^{\infty} \int_{0}^{\infty} e^{-\frac{\epsilon y_2 \beta_{pj} + \beta_{pu} y_1}{\beta_{pu} \beta_{pj} y_1} m_2 -\frac{y_2}{N\beta_{jre}} }   \,dm_2 \,dy_2 }_{I_1}\,dy_1. \nonumber
    \end{align}

    By exploiting \cite[Eq. (3.352.4)]{10.1115/1.3138251}, $I_1$ in (\ref{D1}) is derived as
    \begin{align}
        \label{I1}
        I_1=-\frac{\beta_{pu} y_1}{\epsilon} e^{\frac{\beta_{pu} y_1}{\epsilon N \beta_{pj}\beta_{jre}}} {\rm Ei} \bigg( -\frac{\beta_{pu} y_1}{\epsilon N \beta_{pj} \beta_{jre}}\bigg).
    \end{align}

    Then, substituting (\ref{I1}) into (\ref{D1}), and exploiting Gaussian-Chebyshev quadrature, the closed-form expression of $D^{\Rmnum{1}}$ is obtained, and the proof of Theorem \ref{theorem4} is completed.

\subsection{Proof of Theorem \ref{theorem5}}
The derivation of $C$ is the same as (\ref{C}).
    Let us denote $Y_{\Rmnum{2},1}=\beta_{ure}{\left\lvert \sum_{n = 1}^{N} h_{ur,n} h_{re,n} e^{j\phi_{\Rmnum{1},n}}  \right\rvert }^{2}$ and $Y_{\Rmnum{2},2}= \sum_{n = 1}^{N} \left\lvert {h}_{jr,n} \right\rvert \left\lvert {h}_{re,n} \right\rvert$.
    It can be observed that $Y_{\Rmnum{2},2}$ follows the same distribution in (\ref{PDF3}).
    Hence, we have
    \begin{align}
        \label{D2}
        D^{\Rmnum{2}}=&{\rm Pr}[(1-\tau)\log_2(1+{\rm{SNR}}_{e}^{\Rmnum{1}}) \geq R] \nonumber
        \\
        =&{\rm Pr}\Big[ M_1 Y_{\Rmnum{2},1} \geq \frac{\epsilon \rho_t M_2 \beta_{jre} Y_{\Rmnum{2},2}^2 + \epsilon }{\rho_t } \Big] \nonumber
        \\
        =&\frac{1}{N \beta_{pj} \beta_{ure} \Gamma(v) \varphi^v } \int_{0}^{\infty} e^{-\frac{\epsilon}{\rho_t \beta_{pu} y_1}-\frac{y_1}{N\beta_{ure}}} 
        \\
        & \hspace{-1.5em}\times \underbrace{\int_{0}^{\infty} \int_{0}^{\infty} e^{-\frac{\epsilon  \beta_{pj} \beta_{jre} y_2^2 + \beta_{pu} y_1}{\beta_{pu} \beta_{pj} y_1} m_2  }  y_2^{v-1} e^{-\frac{y_2}{\varphi}}  \,dm_2 \,dy_2 }_{I_2}\,dy_1.\nonumber
    \end{align}

    According to \cite[Eq. (3.389.6)]{10.1115/1.3138251}, $I_2$ in (\ref{D2}) is expressed as
    \begin{align}
        \label{I2}
        &I_2= \frac{\beta_{pj} \Gamma(v-1)}{2} \bigg(\frac{\beta_{pu} y_1}{\epsilon \beta_{pj} \beta_{jre}} \bigg)^{\frac{v}{2}} \bigg[  \exp \left (\alpha_2  +i\frac{(v-2)\pi}{2}  \right )  \nonumber
        \\
        &\times \Gamma\Big(2-v, \alpha_2 \Big)  + \exp\left (-\alpha_2 -i\frac{(v-2)\pi}{2}  \right )  \Gamma\Big(2-v, -\alpha_2 \Big) \bigg],
    \end{align}
    where $\alpha_2=\frac{i}{\varphi}\sqrt{\frac{\beta_{pu}y_1}{\epsilon\beta_{pj}\beta_{jre}} }$.

    By inserting (\ref{I2}) into (\ref{D2}), we use Gaussian-Chebyshev quadrature to yield the closed-form expression of $D^{\Rmnum{2}}$, and the proof of Theorem \ref{theorem5} is completed.

\subsection{Proof of Theorem \ref{theorem6}}
We find that $Y_{\Rmnum{3},1}=Y_{\Rmnum{1},1}$, and $Y_{\Rmnum{3},2}$ follows the same distribution of $\Delta_2$, so that the PDF of $Y_{\Rmnum{3},2}$ is also obtained from (\ref{PDF1}).
    Thus, we have
    \begin{align}
        \label{D3}
        D^{\Rmnum{3}}=&{\rm Pr}[(1-\tau)\log_2(1+{\rm{SNR}}_{e}^{\Rmnum{3}}) \geq R] \nonumber
        \\
        =&{\rm Pr}\Big[ M_1 Y_{\Rmnum{3},1} \geq \frac{\epsilon \rho_t M_2 \beta_{jre} Y_{\Rmnum{3},2} + \epsilon }{\rho_t } \Big] \nonumber
        \\
        =&\frac{1}{N \beta_{pj} \beta_{ure} \Gamma(k_2) \theta_2^{k_2}} \int_{0}^{\infty} e^{-\frac{\epsilon}{\rho_t \beta_{pu} y_1}-\frac{y_1}{N\beta_{ure}}} 
        \\
        & \hspace{-2.5em}\times \underbrace{\int_{0}^{\infty} \int_{0}^{\infty} e^{-\frac{\epsilon  \beta_{pj} \beta_{jre} y_2 + \beta_{pu} y_1}{\beta_{pu} \beta_{pj} y_1} m_2  }  y_2^{k_2-1} e^{-\frac{y_2}{\theta_2}}  \,dm_2 \,dy_2 }_{I_3} \,dy_1. \nonumber
    \end{align}

    Based on \cite[Eq. (3.383.10)]{10.1115/1.3138251}, $I_3$ in (\ref{D3}) is given by
    \begin{align}
        \label{I3}
        I_3=&\beta_{pj} \bigg(\frac{\beta_{pu} y_1}{\epsilon \beta_{pj} \beta_{jre}} \bigg)^{k_2} e^{\frac{\beta_{pu} y_1}{\epsilon \beta_{pj} \beta_{jre} \theta_2}} \Gamma(k_2) \nonumber
        \\
        &\times \Gamma \bigg (1-k_2,\frac{\beta_{pu} y_1}{\epsilon \beta_{pj} \beta_{jre} \theta_2} \bigg).
    \end{align} 
    
    Substituting (\ref{I3}) into (\ref{D3}), and using Gaussian-Chebyshev quadrature, we derive the closed-form expression of $D^{\Rmnum{3}}$, and the proof of Theorem \ref{theorem6} is completed.

\end{appendices}

\vspace{-0.5em}
\tiny
\bibliographystyle{IEEEtran}
\bibliography{IEEEabrv,citation}
\end{document}